\documentclass{article}

\usepackage{PRIMEarxiv}

\usepackage[utf8]{inputenc} 
\usepackage[T1]{fontenc}    
\usepackage{hyperref}       
\usepackage{url}            
\usepackage{booktabs}       
\usepackage{amsfonts,amsthm, amsmath, amssymb}       
\usepackage{nicefrac}       
\usepackage{microtype}      
\usepackage{lipsum}
\usepackage{fancyhdr}       
\usepackage{graphicx}       
\graphicspath{{media/}}     

\usepackage{comment}
\usepackage{booktabs}       
\usepackage{multirow}
\usepackage{float}
\usepackage{subfig}
\usepackage{xcolor}
\usepackage{ulem}

\newtheorem{theorem}{Theorem}[section]
\newtheorem{proposition}[theorem]{Proposition}
\newtheorem{corollary}[theorem]{Corollary}
\newtheorem{lemma}[theorem]{Lemma}
\newtheorem{rmk}[theorem]{Remark}

\newtheorem{example}{Example}

\newcommand{\real}{{\rm I\!R}}
\newcommand{\nat}{{\rm I\!N}}

\title{Some results involving the  \texorpdfstring{$A_{\alpha}$}{alpha}-eigenvalues for graphs and line graphs
\thanks{We would like to thank the National Council for Scientific and Technological Development (CNPq) - Brazil} 
}

\author{
  João Domingos G. da Silva Jr. \\ 
  Departamento de Engenharia de Produção\\
  Centro Federal de Educação Tecnológica do Rio de Janeiro \\
  Rio de Janeiro, Brazil\\
  \texttt{joao.dgomes@gmail.com} \\
  \And
  Carla Silva Oliveira \\ 
  Departamento de Matemática \\
  Escola Nacional de Ci\^encias Estat\'{\i}sticas \\
  Rio de Janeiro, Brazil\\
  \texttt{carla.oliveira@ibge.gov.br} \\
  \And
  Liliana Manuela G. C. da Costa  \\ 
  Departamento de Matemática \\
  Col\'egio Pedro II \\
  Rio de Janeiro, Brazil\\
  \texttt{lmgccosta@gmail.com} \\
}

\begin{document}
\maketitle

\begin{abstract}
Let $G$ be a simple graph with adjacency matrix $A(G)$, signless Laplacian matrix $Q(G)$, degree diagonal matrix $D(G)$ and let $l(G)$ be the line graph of $G$. In 2017, Nikiforov defined the $A_\alpha$-matrix of $G$, $A_\alpha(G)$, as a linear convex combination of $A(G)$ and $D(G)$, the following way, $A_\alpha(G):=\alpha A(G)+(1-\alpha)D(G),$ where $\alpha\in[0,1]$. In this paper, we present some bounds for the eigenvalues of $A_\alpha(G)$ and for the largest and smallest eigenvalues of $A_\alpha(l(G))$. Extremal graphs attaining some of these bounds are characterized.
\end{abstract}


\section{Introduction} \label{intro}

Let $M_{n,m}(\real)$ be the set of $n \times m$ real matrices,  when $m = n$  we use for short $M_n(\real)$. A matrix $M = [m_{ij}]$ is said non-negative ($M \geq 0$) if all its entries, $m_{ij}$, are non-negative, and $M$ is considered  positive ($M > 0$) if all its elements are strictly positive. If $M \in M_n(\real)$, the $M$-characteristic polynomial is defined by $P_M(\lambda) = \vert \lambda I_n - M \vert$ and its roots are called $M$-eigenvalues. We shall index them in non-increasing order and denote by $\lambda_1(M) \geq \ldots \geq \lambda_n(M)$. The collection of $M$-eigenvalues together with their multiplicities is called the $M$-spectrum, denoted by $\sigma(M)$. The largest $M$-eigenvalue, $\lambda_1(M)$, is also called spectral radius. The Rayleigh quotient is defined by $ R(M,x) = \dfrac{x^TMx}{x^Tx},$ for all nonzero vector $x \in \real^n$ . 

Let $G=(V,E)$ be a simple graph such that $\vert V\vert = n$ and $\vert E \vert = m$. For each vertex $v \in V$ the degree of $v$, denoted by $d(v)$, is defined by the number of edges incident to $v$. The minimum degree of $G$, is denoted by $\delta(G) = \min \{ d(v): v \in V \} $ and the maximum degree of $G$ by $\Delta(G) = \max\{d(v): v \in V\}$. The average degree of the neighbors of $v_i \in V$ is $\displaystyle m_i =\frac{1}{d(v_i)} \sum_{v_j \sim v_i}d(v_j) $. The graph $G$ is called $r$-regular if each vertex of $G$ has degree $r$. The graph $G$ is called non-null if it has at least one edge, and $G$ is connected if every pair of distinct vertices of $G$ is joined by a path in $G$. The complement of $G$, denoted by $\overline{G}= (\overline{V},\overline{E})$, is the graph  obtained from $G$ with the same vertex set, $\overline{V} = V$, and $v_iv_j \in \overline{E}$ if and only if $v_iv_j \notin E$. Let $G = (V, E)$ and $H = (W, F)$ be graphs,  if $W \subset V$ and $F \subset E$, then $H$ is a subgraph of $G$. The subgraph denoted by $G - e = (V, E - {e})$ is obtained from $G$ by deleting the edge $e$.

We denote the path with $n$ vertices by $P_n$, the complete graph by $K_n$, the complete bipartite graph with order $n = n_1 + n_2$ by $K_{n_1,n_2}$ and, in particular, the star by $K_{1,n-1}$. A wheel graph of order $n$, $W_n$, is a graph that results of connecting all the vertices of a cycle of order $n-1$ to a single universal vertex known as the hub. The pineapple graph, $K_p^q$, is obtained by appending $q$ pendant edges to a vertex of a complete graph $K_p$ where $q \geq 1$ and $p \geq 3$. The binomial tree\footnote{For more details we suggest \cite{cormen2001introduction}.}, denoted by $BT_k$, is a tree defined recursively as follows: $BT_0$ consists of a single vertex and the binomial tree $BT_k$ consists of two binomial trees $BT_{k-1}$ that are linked together by an edge, as show in Figure \ref{fig::BT_n}.

\begin{figure}[H]
\centering
    \subfloat[Recursively tree.]{\includegraphics[width=0.23
\textwidth]{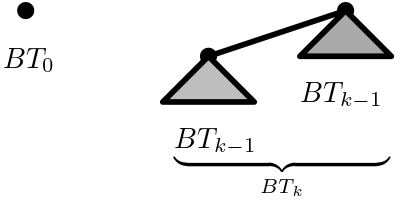}}\qquad\qquad\hspace{10mm}
    \subfloat[The first four binomial trees]{\includegraphics[width=0.45\textwidth]{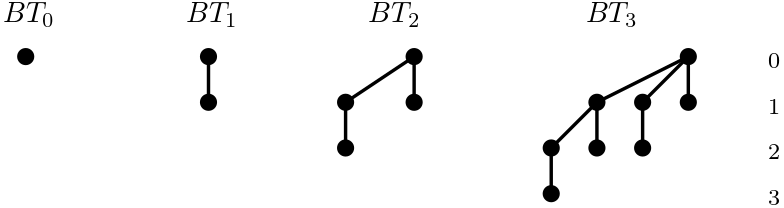}}
    \caption{Binomial tree.}
    \label{fig::BT_n}
\end{figure}

The line graph of $G$, denoted by $l(G)$, is obtained in the following way: each edge in $G$ corresponds to a vertex in $l(G)$, and for two edges in $G$ that share a vertex, make an edge between their corresponding vertices in $l(G)$. From the definition it is important to note that $l(P_{n}) = P_{n-1}$ and $l(K_{1,n-1}) = K_{n-1}$. An example of a line graph can be seen in Figure~\ref{fig::l(G)}.

\begin{figure}[H]
\centering
    \includegraphics[width=0.5\textwidth]{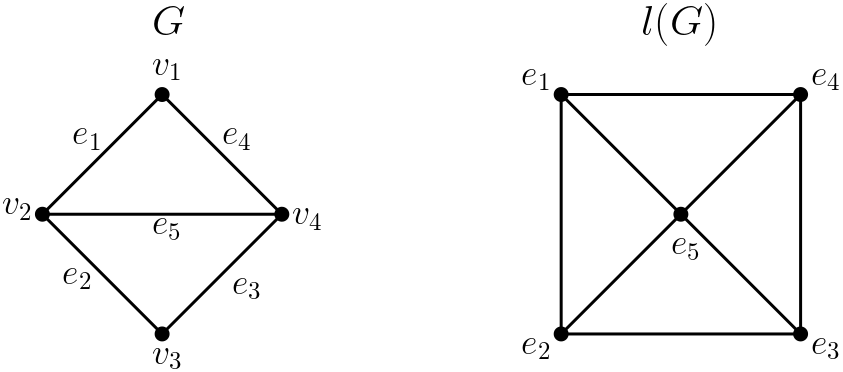}
    \caption{A graph and its line graph.}
    \label{fig::l(G)}
\end{figure}

The first and the second Zagreb indices defined by $\displaystyle Z_1(G) = \sum_{i=1}^n d^2(v_i)$ and $\displaystyle Z_2(G) = \sum_{uv \in E(G)}d(u)d(v)$, respectively, were introduced by Gutman and Trinajestic,~\cite{GUTMAN1972}. The first general Zagreb index is defined by $\displaystyle Z^{(p)}(G) = \sum_{i=1}^n d^{p}(v_i),$ for $p \in \real$, $p \neq 0$ and $p \neq 1$ and it seems to have been first considered by Li et al. in~\cite{Li2004,Li2005}. For $p = 2$, we have the $Z^{(2)}(G) = Z_1(G)$ and the study of its bounds and properties can be found at~\cite{NKMT2003,KinkarCh2003,DAS200457,CIOABA20061959,DasZagreb} and, for $p = 3$, we have $Z^{(3)}(G) = F(G)$, called the forgotten topological index  or F-index, whose study  appears in~\cite{GUTMAN1972,furtula2015,javaid2021}. The general Randi\'{c} index is defined by $\displaystyle R_a(G) = \sum_{uv \in E(G)}\left(d(u)d(v)\right)^a, \text{ where } a \in \real^{*}$ and was introduced by Bollobás and Erd\"{o}s in~\cite{bollobas1998}. The study of its bounds can be found at~\cite{bo2009} and it is not difficult to see that there are close relations between these topological indices, for example, $R_1(G) = Z_2(G)$.

The adjacency matrix of $G$, denoted by $A = A(G) = [a_{ij}]$, is a square and symmetric matrix of order $n$, such that $a_{ij} = 1$ if $v_i$ is adjacent to $v_j$ and $a_{ij} = 0$, otherwise. The incident matrix of $G$, denoted by $B = B(G) = [b_{ij}]$, is a matrix of order $n \times m$ such that $b_{ij} = 1$ if $e_j$ is an incident edge at $v_i$ and $b_{ij} = 0$, otherwise. The degree matrix of $G$, denoted by $D(G)$, is the diagonal matrix that has the degree  of the vertex $v_i$, $d(v_i)$, in the $i^{th}$ position. The matrices $L(G) = D(G) - A(G)$ and $Q(G) = D(G) + A(G)$ are called Laplacian matrix and signless Laplacian matrix, respectively. For simplify the notation, we use $\lambda_i(Q(G)) = q_i$ for all $i = 1, \dots, n$.

In $2017$ Nikiforov, \cite{VN17}, defined for any real $\alpha \in [0,1]$ the convex linear combination $A_\alpha(G)$ of $A(G)$ and $D(G)$ in the following way:
\begin{equation} \label{eq::A_alpha}
    A_\alpha(G) = \alpha D(G) + (1-\alpha)A(G),
\end{equation}
which we call the $A_\alpha$-matrix. From the definition it is easy to see that $A_0(G) = A(G)$, $A_1(G) = D(G)$ and $A_{\frac{1}{2}}(G) = \displaystyle{\frac{1}{2}}Q(G)$. So, obtaining bounds for $A_\alpha$-eigenvalues is an
interesting problem because it contemplates the study of bounds for the adjacency and signless Laplacian matrices.

In this paper, some bounds for the $A_\alpha$-eigenvalues are obtained for a simple graph and for its line graph. We present two lower bounds for $\lambda_1(A_\alpha(G))$ and make a comparison between them. Furthermore, we compare the new bounds with those existing in the literature and presented here and to do this, certain criteria needed to be defined. Specifically, we opted to evaluate bounds that pertain to identical extremal graphs.  

This paper is organized as follows: in Section \ref{sec::preliminaries}, we introduce some definitions and results required to prove the main results; in Section \ref{sec::main}, we show the main results, together with some comparisons of the obtained bounds.

\section{Preliminaries}\label{sec::preliminaries}

In this section we present some results that will be useful to prove the main results of the paper. We start with the Theorem of Weyl and So, which inequalities involve eigenvalues of sums of Hermitian matrices. 

\begin{theorem}[Weyl] \label{theo::weyl}
\cite{horn2013matrix} Let $A, B \in M_n({\rm I\!R})$ be symmetric and let the spectrum of $A$, $B$, and $A + B$ be $\sigma(A) = \{\lambda_1(A), \ldots, \lambda_n(A)\}$, $\sigma(B) = \{\lambda_1(B), \ldots, \lambda_n(B)\}$ and $\sigma(A+B) = \{\lambda_1(A+B), \ldots, \lambda_n(A+B)\}$, respectively. Then,

\begin{equation}\label{weyl_upper_bound}
\lambda_{i+j-1}(A+B) \leq \lambda_i(A) + \lambda_j(B), \ \ j = 1, \ldots , n-i+1
\end{equation}
for each $i = 1, \ldots, n$, with equality for some pair $i,j$ if and only if there is a nonzero vector $x$ such that $Ax = \lambda_i x$, $Bx = \lambda_j x$ and $(A+B)x = \lambda_{i+j-1} x$.
Also,	

\begin{equation}\label{weyl_lower_bound}
\lambda_i(A) + \lambda_j(B) \leq \lambda_{i + j -n}(A+B), \ \ j = i, \ldots, n
\end{equation}
for each $i = 1, \ldots, n$, with equality for some pair $i,j$ if and only if there is a nonzero vector $x$ such that $Ax = \lambda_i x$, $Bx = \lambda_j x$ e $(A+B)x = \lambda_{i+j-n} x$. If $A$ and $B$ have no common eigenvector, then the inequalities in (\ref{weyl_upper_bound}) and (\ref{weyl_lower_bound}) are strict.
\end{theorem}

As consequence of Theorem \ref{theo::weyl}, we have Corollary \ref{cor::weyl}.

\begin{corollary} \label{cor::weyl}
\cite{horn2013matrix} Let be $A, B \in M_n({\rm I\!R})$ symmetric. Then,

\begin{equation} \label{weyl_inequalities}
\lambda_i(A) + \lambda_n(B) \leq \lambda_i(A+B) \leq \lambda_i(A) + \lambda_1(B),
\end{equation}
with $i = 1, \ldots, n$. Equality in the upper bound holds if and only if there is a nonzero vector $x$  that is eigenvector of $A,B$ and $ A+B$ with corresponding eigenvalues $\lambda_i$, $\lambda_1$ and $\lambda_i$, respectively. Analogously, equality in the lower bound holds if and only if there is nonzero vector $x$  that is eigenvector of $A,B$ and $ A+B$ with corresponding eigenvalues $\lambda_i$, $\lambda_n$ and $\lambda_i$, respectively.
\end{corollary}

Theorem \ref{theo::interlacing} relates the spectra of a graph and its subgraph, Theorem \ref{theo::rayleigh} and Proposition \ref{prop::rayleigh_alpha} show the Rayleigh theorem and its adaptation to $A_\alpha$-matrix.

\begin{theorem} \label{theo::interlacing}
\cite{cvetkovic2009introduction} Let $G$ be a graph with $n$ vertices and eigenvalues $\lambda_1(A(G)) \geq \ldots \geq \lambda_n(A(G))$, and let $H$ an induced subgraph of $G$ with $s$ vertices. If the eigenvalues of $H$ are $\lambda_1(A(H)) \geq \ldots \geq \lambda_s(A(H))$ then $\lambda_{n-s+i}(A(G)) \leq \lambda_i(A(H)) \leq \lambda_i(A(G)), \ \ \forall i = 1, \ldots , s$.
\end{theorem}

\begin{theorem} \label{theo::rayleigh}
\cite{horn2013matrix} Let $A \in M_n(\real)$ symmetric with eigenvalues $\lambda_1 \geq \lambda_2 \geq \ldots \geq \lambda_n$. Then,
$$ \lambda_1 = \max_{x \neq 0} R(A,x) \ \ \textit{and} \ \ \lambda_n = \min_{x \neq 0} R(A,x)$$
\end{theorem}

\begin{proposition} \label{prop::rayleigh_alpha}
\cite{VN17} If $\alpha \in [0,1]$ and $G$ is a graph of order $n$, then
\begin{equation}
\lambda_1(A_\alpha(G)) = \max_{\vert x \vert = 1} x^TA_\alpha(G) x \text{ and }  \lambda_n(A_\alpha(G)) = \min_{\lvert x \rvert = 1} x^TA_\alpha(G)x.
\end{equation}
Furthermore, if $x$ is a unit vector, then $\lambda_1(A_\alpha(G)) = x^TA_\alpha(G)x$ if and only if $x$ is an eigenvector of $\lambda_1(A_\alpha(G))$, and $\lambda_n(A_\alpha(G)) = x^TA_\alpha(G)x$ if and only if $x$ is an eigenvector of $\lambda_n(A_\alpha(G))$.
\end{proposition}

The next result shows a lower bound for the largest eigenvalue of a non-negative matrix.

\begin{lemma} \label{lemma::LowerBound_SpectralRadius}
\cite{KOLOTILINA1993133, CHEN2010908} Let $B = (b_{ij})$ be a non-negative $n \times n$ matrix with $n \geq 2$, $\lambda_1(B)$ be the largest eigenvalue of $B$, and set $\displaystyle \theta = \min_{1 \leq i \leq n} \{b_{ii}\}$. Then
\begin{equation}
  \displaystyle \lambda_1(B) \geq \max_{i} \left\{ \dfrac{b_{ii} + \theta}{2} + \sqrt{\dfrac{(b_{ii} - \theta)^2}{4} + \sum_{i \neq j}b_{ij}b_{ji} }\right\}.
\end{equation}
Moreover, if $B$ is irreducible with $n \geq 3$, and at least two rows (two columns) of $B$ contain more than one nonzero off-diagonal entry, then inequality is strict.
\end{lemma}

\begin{lemma} \label{lemma::rowsums}
\cite{ELLINGHAM200045} Let $G$ be a connected graph with $n$ vertices and $A(G) = A$ its adjacency matrix. Let $P(x)$ be any polynomial function and $S_v(P(A)$ be the row sums of $P(A)$ corresponding to each vertex $v$. Then
$$\min{S_v(P(A))} \leq P(\lambda_1(A)) \leq \max{S_v(P(A))}.$$
Moreover, equality holds if and only if the row sums of $P(A)$ are all equal.
\end{lemma}

Bounds for the first Zagreb Index, the F-index and the general Randi\'{c} index are presented in the next results.

\begin{lemma} \label{theo::sharp_DAS}
	\cite{KinkarCh2003} Let $G$ be a simple graph with $n$ vertices and $m$ edges. Let $\delta$ and $\Delta$ be the minimum  and the maximum degree of $G$, respectively. Then, for $n \geq 3$, $\displaystyle Z_1(G) \geq \Delta^2 + \delta^2 + \frac{(2m - \Delta - \delta)^2}{n-2}$. Furthermore, equality occurs if and only if $d_2 =  \dots = d_{n-1}$.
\end{lemma}

\begin{lemma} \label{theo::sharp_DAS_upper}
\cite{KinkarCh2003} Let $G$ be a connected graph with $n$ vertices and $m$ edges. Let $\delta$ be the minimum degree of $G$.  Then, $\displaystyle Z_1(G) \leq 2mn -n(n-1)\delta + 2m(\delta-1)$. Moreover, the equality holds if and only if $G$ is a star graph or a regular graph.
\end{lemma}

\begin{proposition} \label{prop::forgotten}
    \cite{furtula2015} Let $G$ be a graph with $m$ edges, whose first Zagreb index is $Z_1(G)$. Then,
    
    \begin{equation*}
        F(G) \geq \dfrac{Z_1(G)^2}{2m}
    \end{equation*}
Equality is attained in the case of regular graphs.
\end{proposition}

\begin{theorem} \label{theo::randic}
     \cite{bo2009} Let $G$ be a graph with $n$ vertices and $m \geq 1$ edges. Then for $a \geq 1$,
     \begin{equation*}
         R_a \geq 4^an^{-2a}m^{1+2a}
     \end{equation*}
with equality if and only if $G$ is a regular graph.
\end{theorem}

Theorem \ref{theo::LineGraphs_properties} and Theorem \ref{theo::degree_relation} show a relation between the vertices degree of the graph and its line graph. Theorem \ref{theo::lowerBound_LG}, Lemma \ref{lemma::CompleteLineGraph} and Proposition \ref{prop::SubgraphLineGraph}, present some results involving line graphs.

\begin{theorem} \label{theo::LineGraphs_properties}
\cite{beineke2021line} Let $G$ be a non-null graph such that $V(G) = \{v_1, \ldots, v_n\}$ and $m$ edges. Then
\begin{itemize}
\item[(a)] $l(G)$ has $m$ vertices and $\displaystyle \frac{1}{2} \sum_{i=1}^n d^2(v_i)-m$ edges.
\item[(b)] The degree of a vertex $e = v_iv_j$ in $l(G)$ is $d(e) = d(v_i) + d(v_j) - 2$.
\end{itemize}
\end{theorem}

\begin{theorem} \label{theo::degree_relation}
\cite{beineke2021line} Let $G$ be a graph with at least one edge. Then,
\begin{itemize}
    \item [(a)] $\delta(l(G)) \geq 2\delta(G) - 2$ with equality if and only if $G$ has two adjacent vertices of degree $\delta(G)$.
    \item[(b)] $\Delta(l(G)) \leq 2\Delta(G) - 2$ with equality if and only if $G$ has two adjacent vertices of degree $\Delta(G)$.
\end{itemize}
\end{theorem}

\begin{theorem} \label{theo::lowerBound_LG}
\cite{beineke2021line} If $\lambda_m(l(G))$ is the smallest eigenvalue of $A(l(G))$, then $-2 \leq \lambda_m(l(G)).$
\end{theorem}

\begin{lemma} \label{lemma::CompleteLineGraph}
\cite{chen2002} Let $G$ be a connected graph with $n$ vertices, then $l(G)$ is the complete graph if and only if $G$ is either $K_{1,n-1}$ or $K_3$.
\end{lemma}

\begin{proposition} \label{prop::SubgraphLineGraph}
\cite{beineke2021line} If $H$ is a non-null subgraph of $G$, then $l(H)$ is an induced subgraph of $l(G)$.
\end{proposition}

As consequence of Proposition \ref{prop::SubgraphLineGraph} and Theorem \ref{theo::interlacing} we have the Corollary \ref{cor::InterlacingLineGraph}.

\begin{corollary} \label{cor::InterlacingLineGraph}
Let $G$ be a graph with $n$ vertices and $m \neq 0$ edges. Then
$$\lambda_i(A_\alpha(l(G))) \geq \lambda_i(A_\alpha(l(G - e))) \geq \lambda_{i+1}(A_\alpha(l(G))),$$
$\forall i = 0, \ldots,m-1$ and $\alpha \in [0,1]$.
\end{corollary}

Some known results in the literature about the incident matrix, Laplacian matrix and signless Laplacian matrix are presented below.

\begin{lemma} \label{lemma::incident_BTB}
\cite{cvetkovic2009introduction} Let $G$ be a graph with $m$ edges and $B = B(G)$ the incident matrix of $G$. Then $B^TB = 2I_m + A(l(G))$.
\end{lemma}

\begin{rmk} \label{rmk::incident_BTB_alpha}
From Lemma \ref{lemma::incident_BTB} follows that

\begin{equation} \label{eq::BTB}
    (1-\alpha)B^TB = (1-\alpha)A(l(G)) + 2(1-\alpha)I_m = A_\alpha(l(G)) -\alpha D(l(G)) + 2(1-\alpha)I_m.
\end{equation}
Taking $U = -\alpha D(l(G)) + 2(1-\alpha)I_m$ and substituting in equation (\ref{eq::BTB}), we have
\begin{equation} \label{eq::BTB_1}
    (1-\alpha)B^TB = A_\alpha(l(G)) + U,
\end{equation}
where $U = [u_{ij}]$ is a diagonal matrix of order $m$. From Theorem \ref{theo::LineGraphs_properties}, for all $k$ such that $1 \leq k \leq m$ and $e_k = v_iv_j$  we have that $u_{kk} = -\alpha d(e_k) + 2 -2\alpha = -\alpha(d(v_i)+ d(v_j) - 2) + 2 - 2\alpha = 2 - \alpha(d(v_i) + d(v_j))$.
\end{rmk}

\begin{lemma} \label{lemma::incident_BBT}
\cite{cvetkovic2009introduction} Let $G$ be a graph with $n$ vertices and $m$ edges. Consider $B$ and $D(G)$ the incident and the degree matrix of $G$, respectively. Then $BB^T = D(G) + A(G) = Q(G)$.
\end{lemma}

\begin{rmk}\label{rmk::incident_BBT_alpha}
From Lemma \ref{lemma::incident_BBT} follows that
\begin{equation} \label{eq::BBT}
(1-\alpha)BB^T = (1-\alpha)A(G) + (1-\alpha)D(G) = A_\alpha(G) -\alpha D(G) + (1-\alpha)D(G) = A_\alpha(G) + (1-2\alpha)D(G)
\end{equation}
\end{rmk}

\begin{proposition} \label{prop::signlessLeasteigen}
\cite{CVETKOVIC2007155} The least eigenvalue of the signless Laplacian of a connected graph is equal to $0$ if and only if the graph is bipartite. In this case $0$ is a simple eigenvalue.
\end{proposition}

\begin{proposition} \label{prop::Cvetkovic2007EIGENVALUEBF}
    \cite{Cvetkovic2007EIGENVALUEBF} The matrices $L(G)$ and $Q(G)$ have the same characteristic polynomial if and only if $G$ is a bipartite graph.
\end{proposition}

Lemma \ref{lemma::eigeneq_RegularGraphs} shows a linear correspondence between the eigenvalues of $A_\alpha(G)$ and $A(G)$ and Proposition \ref{prop::Perron_alpha} is an adaptation of Perron-Frobenius's Theorem for $A_\alpha$-matrix.

\begin{lemma} \label{lemma::eigeneq_RegularGraphs}
\cite{VN17} If $\alpha \in [0,1]$ and $k = 1, \ldots, n$ and $G$ is a $r$-regular graph of order $n$, then there exists a linear correspondence between the eigenvalues of $A_\alpha(G)$ and $A(G)$, the following way

\begin{equation} \label{eq::autoequation}
\lambda_k(A_\alpha(G)) = \alpha r + (1-\alpha)\lambda_k(A(G)).
\end{equation}
In particular, $\lambda_1(A_\alpha(G)) = r, \ \ \forall \alpha \in [0,1]$.
\end{lemma}

\begin{proposition} \label{prop::Perron_alpha}
\cite{VN17} Let $\alpha \in [0,1)$, $G$ be a graph and $x$ be a non- negative eigenvector of $\lambda_1(A_\alpha(G))$.
\begin{itemize}
\item [(i)] If $G$ is connected, then $x$ is positive and unique minus scalar;
\item [(ii)] If $G$ is disconnected and $P$ is the set of vertices with positive entries of $x$, then the subgraph induced by $P$ is a union of  components $H$ of $G$ with $\lambda_1(A_\alpha(H)) = \lambda_1(A_\alpha(G))$;
\item [(iii)] If $G$ is connected and $\mu$ is an eigenvalue of $A_\alpha(G)$ with a non-negative eigenvector, then $\mu = \lambda_1(A_\alpha(G))$;
\item [(iv)] If $G$ is connected, and $H$ is a proper subgraph of $G$, then $\lambda_1(A_\alpha(H))< \lambda_1(A_\alpha(G))$.
\end{itemize}
\end{proposition}

Proposition \ref{prop::complete_graph_spectrum} and Corollary \ref{cor::StarLineGraph} show the $A_\alpha$-spectrum of $K_n$ and $l(K_{1,n-1}).$

\begin{proposition} \label{prop::complete_graph_spectrum}
\cite{VN17} The eigenvalues of $A_\alpha(K_n)$ are $\lambda_1(A_\alpha(K_n)) = n-1$ and $\lambda_k(A_\alpha(K_n)) = \alpha n -1 \text{  for } 2 \leq k \leq n$.
\end{proposition}

\begin{corollary} \label{cor::StarLineGraph}
    Let $G \cong K_{1,n-1}$ and $\alpha \in [0,1]$. Then $\sigma(A_\alpha(l(G))) = \left\{n-2^{(1)}, (n-1)\alpha - 1^{(n-2)} \right\}.$
\end{corollary}
\begin{proof}
    From Lemma \ref{lemma::CompleteLineGraph} we know that $l(G)$ is a complete graph with $n-1$ vertices and from Proposition \ref{prop::complete_graph_spectrum} the result follows.
\end{proof}

Theorem \ref{theo::linegraph} provides relations between $P_{A_\alpha(l(G))}$ and $P_{A_\alpha(G)}$,  and between $P_{A_\alpha(l(G))}$ and $P_{A(G)}$. Corollary \ref{cor::complete_graph} obtains the $A_\alpha$-spectrum of the $l(K_n)$ and Corollary \ref{cor::pol_signless} shows a relation between  $P_{A_\alpha(l(G))}$ and $P_{Q(G)}$ when $G$ is $r$-regular.

\begin{theorem} \label{theo::linegraph}
	\cite{joao2022} Let $G$ be a $r$-regular graph with $n$ vertices and $m$ edges such that $r \geq 2$ and $\alpha \in [0,1)$. Then
	\begin{equation*} \label{eq1::linegraph}
	    P_{A_\alpha(l(G))}(\lambda) = (\lambda - 2r \alpha + 2)^{m-n}P_{A_\alpha(G)}(\lambda - r + 2)
	\end{equation*}
	and
	\begin{equation*} \label{eq2::linegraph}
	    P_{A_\alpha(l(G))}(\lambda) = (\lambda - 2r \alpha + 2)^{m-n}(1-\alpha)^n P_{A(G)}\left( \dfrac{\lambda - r(\alpha+1) + 2}{1-\alpha} \right)
	\end{equation*}
\end{theorem}

\begin{corollary} \label{cor::complete_graph}
    Let $\alpha \in [0,1]$. Then $\sigma(A_\alpha(l(K_n))) = \Biggl\{ 2n-4, (n(\alpha + 1) - 4)^{(n-1)}$, $(2\alpha(n-1) - 2)^{(\frac{n(n-3)}{2})} \Biggr\}.$
\end{corollary}
\begin{proof}
From Proposition \ref{prop::complete_graph_spectrum} and Theorem \ref{theo::linegraph} the result follows.
\end{proof}

\begin{corollary} \label{cor::pol_signless}
    Let $G$ be a $r$-regular graph with $n$ vertices, $m \neq 0$ edges and $\alpha \in [0,1)$. Then
    \begin{equation*} \label{eq2::pol_linegraph}
    P_{A_\alpha(l(G))}(\lambda) = (\lambda - 2r \alpha + 2)^{m-n}(1-\alpha)^n P_{Q(G)}\left( \dfrac{\lambda - 2r\alpha + 2}{1-\alpha} \right)
    \end{equation*}
\end{corollary}
\begin{proof}
    From Theorem \ref{theo::linegraph} and Lemma \ref{lemma::incident_BBT} follows that
    \begin{align*}
        P_{A_\alpha(l(G))}(\lambda) &= (\lambda -2r \alpha + 2)^{m-n} \vert(\lambda -2r \alpha + 2)I_n - (1-\alpha)BB^T\vert \\
        &= (\lambda - 2r \alpha + 2)^{m-n}(1-\alpha)^n P_{Q(G)}\left( \dfrac{\lambda - 2r\alpha + 2}{1-\alpha} \right).
    \end{align*}
\end{proof}

\begin{rmk} \label{rmk::signlessLaplacian_alpha}
Let $G$ a $r$-regular graph with $\sigma(Q(G)) = \{q_1, \ldots, q_n\}$ and $m > n$. From Corollary \ref{cor::pol_signless}, we have $\alpha,$ $(2r\alpha-2)$ and $\alpha(2r-q_i) + q_i-2$ for $i=1, \dots, n$ belong to $\sigma(A_\alpha(l(G)))$.
\end{rmk}

\begin{example} \label{example::complete_graph}
Consider $G \cong K_n$.  We know that $\sigma(Q(K_n)) = \{2n-2,  n - 2^{(n-1)} \}$, see~\cite{CARDOSO2018325,cvetkovic2004}. So, from Corollary \ref{cor::pol_signless} we have  $\sigma(A_\alpha(l(K_n))) = \Biggl\{ 2n-4, n(\alpha + 1) - 4^{(n-1)}$, $2\alpha(n-1) - 2^{(\frac{n(n-3)}{2})} \Biggr\}$, which can be seen in~\cite{joao2022}.
\end{example}

\begin{example} \label{example::bipartite_graph}
Consider $G \cong K_{n,n}$. From Proposition \ref{prop::Cvetkovic2007EIGENVALUEBF}, $\sigma(L(K_{n,n}))= \sigma(Q(K_{n,n}))$. Moreover, $\sigma(L(K_{n,n})) = \{2n, n^{(2n-2)},0 \}$, which can be see in~\cite{Malathy2017BOUNDSFL, cvetkovic2004}. So from Corollary \ref{cor::pol_signless}, we obtain $\sigma(A_\alpha(l(K_{n,n}))) = \Biggl\{ 2n-2, n(\alpha + 1) - 2^{(2n-2)}$, $2\alpha n - 2^{(n^2-2n+1)} \Biggr\}$.
\end{example}

The following results are bounds for the $A_\alpha$-eigenvalues. In~\cite{VN17}, Proposition \ref{prop::lower_bound_min_and_max} was introduced without extremal graphs so, we rewrite it and introduce its extremal graph.

\begin{proposition} \label{prop::lowerBound_niki}
    \cite{VN17} If $G$ is a graph with maximum degree $\Delta$, then
    \begin{equation} \label{eq::niki_lowerBound_extremalStar}
        \lambda_1(A_\alpha(G)) \geq \dfrac{1}{2}\left( \alpha(\Delta + 1) + \sqrt{\alpha^2(\Delta +1)^2 + 4 \Delta(1-2\alpha)}\right)
    \end{equation}
    If $G$ is connected, equality holds if and only if $G \cong K_{1,\Delta}$.
\end{proposition}

\begin{proposition} \label{prop::lower_bound_min_and_max}
   \cite{VN17}  Let $G$ be a graph with $n$ vertices and $\alpha \in [0,1]$. Then, 
    \begin{equation} \label{eq::LowerUpperBound_Wang}
        \min_{v_i \in V}\left\{ \sqrt{\alpha d^2(v_i) + (1-\alpha)\sum_{v_j \sim v_i}d(v_j)} \right\} \leq \lambda_1(A_\alpha(G)) \leq \max_{v_i \in V}\left\{ \sqrt{\alpha d^2(v_i) + (1-\alpha)\sum_{v_j \sim v_i}d(v_j)} \right\}
    \end{equation}
    Moreover, the equalities holds if and only if $G$ is regular.
\end{proposition}
\begin{proof}
    We have already seen that $S_{v_i}(A_\alpha(G)) = d(v_i)$ and from \cite{VN17} we have that $\displaystyle S_{v_i}(A_\alpha^2(G)) = \alpha d^2(v_i) + (1-\alpha)\sum_{v_j \sim v_i}d(v_j)$.

    From Lemma \ref{lemma::rowsums} follows that
    \begin{equation*}
        \min_{v_i \in V} \{ S_{v_i}(A_\alpha^2(G)) \} \leq \lambda_1^2(A_\alpha(G)) \leq \max_{v_i \in V} \{ S_{v_i}(A_\alpha^2(G))\}
    \end{equation*}
    and then the result follows.

    Now, suppose initially that $G$ is a $r$-regular graph. Hence, $d(v_i) = r, \forall v_i \in V$. Replacing in (\ref{eq::LowerUpperBound_Wang}) we obtain $\lambda_1(A_\alpha(G)) = r$. Conversely, if both equalities hold, we have that all row sums are equal and then we can conclude that $G$ is regular.
\end{proof}

\begin{proposition} \label{prop::second_alpha}
    \cite{Zhang2019} Let $G$ be a graph with $n$ vertices. If $G \ncong K_n$, then $\lambda_2(A_\alpha(G)) \geq 0.$
\end{proposition}

\begin{proposition} \label{prop::upper_largesteigen_Pn} \cite{NIKIFOROV2017286} The largest eigenvalue of $A_\alpha(P_n)$ satisfies
\begin{equation*}
    \lambda_1(A_\alpha(P_n)) \leq \begin{cases}
    2\alpha + 2(1-\alpha)\cos{\left(\dfrac{\pi}{n+1}\right)}, \ \ \text{if} \ \ 0 \leq \alpha < \dfrac{1}{2};\\
    2\alpha + 2(1-\alpha)\cos{\left(\dfrac{\pi}{n}\right)}, \ \ \text{if} \ \ \dfrac{1}{2} \leq \alpha \leq 1. 
    \end{cases}
\end{equation*}
Equality holds if and only if $\alpha =0$, $\alpha = \dfrac{1}{2}$ or $\alpha = 1$.
\end{proposition}

\section{Main Results} \label{sec::main}

In this section we show the main results of this paper that involve bounds for some $A_\alpha$-eigenvalues of graphs and line graphs and, when is possible, we exhibited extremal graphs. Moreover, comparisons between some bounds are presented.

\subsection{Some Bounds for \texorpdfstring{$\mathbf{A_\alpha}$}{text2}-eigenvalues of Graphs} \label{subsection::A_alpha-Graphs}


\begin{theorem} \label{theo::LowerBound_alpha}
    Let $G$ be a graph of order $n \geq 2$, $\Delta$ and $\delta$ be the maximum degree and the minimum degree of $G$, respectively, and $\alpha \in [0,1]$ . Then
    \begin{equation} \label{eq::lowerbound}
        \lambda_1(A_\alpha(G)) \geq \dfrac{\alpha(\Delta + \delta) + \sqrt{\alpha^2(\Delta - \delta)^2 + 4(1-\alpha)^2 \Delta}}{2}
    \end{equation}
    If $G$ is connected, the equality holds if and only if $G \cong K_{1,n-1}$.
\end{theorem}
\begin{proof}
From Lemma \ref{lemma::LowerBound_SpectralRadius} we have that 
\begin{align*}
    \displaystyle \lambda_1(A_\alpha(G))) & \geq \max_{i} \left\{ \dfrac{A_\alpha(G)_{ii} + \alpha \delta}{2} + \sqrt{\dfrac{(A_\alpha(G)_{ii} - \alpha \delta)^2}{4} + \sum_{i \neq j}A_\alpha(G)_{ij}A_\alpha(G)_{ji} }\right\} \\
    & =  \dfrac{\alpha \Delta + \alpha \delta}{2} + \sqrt{\dfrac{(\alpha \Delta - \alpha \delta)^2}{4} + \max_{i} \left\{ \sum_{i \neq j}A_\alpha(G)_{ij}A_\alpha(G)_{ji} \right\} } \\
    & =  \dfrac{\alpha (\Delta + \delta) + \sqrt{\alpha^2 (\Delta - \delta)^2 + \displaystyle 4\max_{i} \left\{ \sum_{i \neq j}A_\alpha(G)_{ij}A_\alpha(G)_{ji} \right\} }}{2} 
\end{align*}
As $A_\alpha(G) = \alpha D(G) + (1-\alpha)A(G)$,  follows that
\begin{equation*}
    \displaystyle \sum_{i \neq j}A_\alpha(G)_{ij}A_\alpha(G)_{ji} = (1-\alpha)^2 \sum_{i \neq j} a_{ij}a_{ji}
\end{equation*}
and then
\begin{equation*}
   \lambda_1(A_\alpha(G))) \geq  \dfrac{\alpha (\Delta + \delta) + \sqrt{\alpha^2 (\Delta - \delta)^2 + 4(1-\alpha)^2 \Delta}}{2}.
\end{equation*}

If $G \cong K_{1,n-1}$ we have that $\Delta = n-1$ and $\delta = 1$, so from~\cite{VN17} we have the equality. Now suppose that equality in (\ref{eq::lowerbound}) holds. Since $G$ is connected, $A_\alpha(G)$ is irreducible. By the equality condition  in Lemma \ref{lemma::LowerBound_SpectralRadius}, there exists only one row (column) of $A_\alpha(G)$ containing more than one nonzero off-diagonal entry. Then there exists only a vertex $v$ with $d(v) \geq  2$. So, $G \cong K_{1,n-1}$.
\end{proof}

\begin{corollary} \label{cor::LowerBoundEquality}
    Let $G$ be a graph of order $n \geq 2$, $\Delta$ the maximum degree, $\delta = 1$ and $\alpha \in [0,1).$ Then the lower bound presented in \eqref{eq::lowerbound} and \eqref{eq::niki_lowerBound_extremalStar} are equal.
\end{corollary}
\begin{proof}
    Taking $\delta = 1$ in~\eqref{eq::lowerbound} we have
    \begin{align*}
        \lambda_1(A_\alpha(G)) &\geq \dfrac{\alpha(\Delta + 1) + \sqrt{\alpha^2(\Delta - 1)^2 + 4(1-\alpha)^2 \Delta}}{2}\\
        & = \dfrac{\alpha(\Delta + 1) + \sqrt{\alpha^2\Delta^2 - 2\Delta\alpha^2 + \alpha^2 + 4\Delta -8\alpha\Delta + 4\alpha^2\Delta}}{2}\\
        & = \dfrac{\alpha(\Delta + 1) + \sqrt{\alpha^2(\Delta +1)^2 + 4\Delta(1-2\alpha)}}{2}
    \end{align*}
\end{proof}

\begin{theorem}
    Let $\alpha \in [0,1]$ and $G$ be a graph with $m \neq 0$ edges, $n \geq 3$ vertices, $\Delta$ and $\delta$ be the maximum and minimum degrees, respectively. Then,
    \begin{equation} \label{eq::lower_bound2}
        \displaystyle \lambda_1(A_\alpha(G)) \geq \alpha \dfrac{(\Delta^2+\delta^2)(n-2) + (2m-\Delta -\delta)^2}{2m(n-2)} + (1-\alpha) \dfrac{8m^3}{n^2(2\delta m + (n-1)(2m - n \delta))}
    \end{equation}
    The equality occurs if $G$ is a regular graph.
\end{theorem}
\begin{proof}
    From Proposition \ref{prop::rayleigh_alpha} we know that there exists an eigenvector $x \in \real^n$ associated with $\lambda_1(A_\alpha(G))$ that satisfies $\displaystyle \lambda_1(A_\alpha(G)) = \max_{x \in \real^n} \frac{x^T A_\alpha(G) x}{x^Tx}$. So  for all $y \neq kx$, where  $\; k \in \real$, we have
    \begin{align*}
        \lambda_1(A_\alpha(G)) &\geq  \frac{y^T A_\alpha(G) y}{y^Ty} = \frac{y^T(\alpha D(G) + (1-\alpha)A(G)) y}{y^Ty} \\
        &=\frac{\alpha y^T D(G)y + (1-\alpha)y^TA(G) y}{y^Ty}.
    \end{align*}
    Taking $y = (d(v_1), d(v_2), \ldots, d(v_n)) = (d_1, d_2, \ldots, d_n)$, we have
    \begin{align*}
        \lambda_1(A_\alpha(G)) &\geq \dfrac{\alpha \displaystyle \sum_{i=1}^n d_{i}^3+ (1-\alpha)y^TA(G)) y}{\displaystyle \sum_{i=1}^n d_{i}^2}=\dfrac{\alpha \displaystyle \sum_{i=1}^n d_{i}^3}{\displaystyle \sum_{i=1}^n d_{i}^2} + \dfrac{(1-\alpha) \displaystyle 2\sum_{v_i \sim v_j} d_{i}d_{j}}{\displaystyle \sum_{i=1}^n d_{i}^2} \\
        & = \alpha \dfrac{\displaystyle \sum_{i=1}^n d_{i}^3}{\displaystyle \sum_{i=1}^n d_{i}^2} + (1-\alpha)\dfrac{\displaystyle  2\sum_{v_i \sim v_j} d_{i}d_{j}}{\displaystyle \sum_{i=1}^n d_{i}^2} = \alpha \dfrac{F(G)}{Z_1(G)} + (1-\alpha)\dfrac{2R_1(G)}{Z_1(G)}.
    \end{align*}
    From Proposition \ref{prop::forgotten} and Theorem \ref{theo::randic} follows that
    $$\lambda_1(A_\alpha(G)) \geq \alpha \dfrac{\dfrac{Z_1^2(G)}{2m}}{Z_1(G)} + (1-\alpha)\dfrac{8\dfrac{m^3}{n^2}}{Z_1(G)} = \alpha \dfrac{Z_1(G)}{2m} + (1-\alpha)\dfrac{8m^3}{n^2Z_1(G)}.$$
    Using Lemmas \ref{theo::sharp_DAS} and \ref{theo::sharp_DAS_upper} with some algebraic manipulation the result follows.

    To prove the equality suppose that $G$ is $r$-regular graph. From Lemma \ref{lemma::eigeneq_RegularGraphs} we know that $\lambda_1(A_\alpha(G)) = r$ and moreover that $\Delta = \delta = r$ and $m = \dfrac{nr}{2}$. Then,
    \begin{align*}
         & \alpha \dfrac{n((\Delta + \delta)^2 -4m(\Delta + \delta-m) +(\Delta^2+\delta^2)(n-2))}{2m} + (1-\alpha) \dfrac{8m^3}{n^2(2\delta m - (n-1)(2m - n \delta))} = \\
         &  \dfrac{n((r + r)^2 -4\dfrac{nr}{2}\left(r + r-\dfrac{nr}{2}\right) +(r^2+r^2)(n-2))}{2\dfrac{nr}{2}} + (1-\alpha) \dfrac{8\left(\dfrac{nr}{2}\right)^3}{n^2\left(2r\dfrac{nr}{2} - (n-1)\left(2\dfrac{nr}{2} - nr\right)\right)} = \\
         & \alpha r - \dfrac{2r\alpha}{n} + \dfrac{2r^2\alpha}{nr}-\alpha r + r = r
    \end{align*}
\end{proof}

\begin{theorem} \label{theo::JCL_upperbound}
     Let $G$ be a graph with $n$ vertices, $m \neq 0$ edges, $\Delta(G)$ its maximum degree, $\delta(G)$ its minimum degree and $\alpha \in [0,1]$. Then, 
    \begin{equation} \label{eq::UpperBound_JCL}
        \lambda_1(A_\alpha(G)) \leq \sqrt{\alpha \Delta^2 + (1 - \alpha)(\Delta(\delta - 1) -\delta(n-1) + 2m) }.
    \end{equation}
    Moreover, equality holds if and only if $G$ is regular.
\end{theorem}

\begin{proof}
    Let $S_v(M)$ the row sums of $M$ corresponding to each vertex $v$. Since $A_\alpha(G) = \alpha D(G) + (1-\alpha)A(G)$, we have that $S_v(A_\alpha(G)) = d(v)$ and $\displaystyle S_v(A(G)D(G)) = S_v(A^2(G)) = \sum_{u \sim v}d(u) = 2m-d(v)-\sum_{\substack{u \nsim v\\ u\neq v}}d(u)$. Then,
    \begin{align*}
        S_v(A^2_\alpha(G)) &= S_v\left((\alpha D(G) + (1-\alpha)A(G))^2\right) = S_v\left(\alpha D(G)A_\alpha(G) + \alpha(1-\alpha)A(G)D(G) + (1-\alpha)^2A^2(G)\right)\\
        & = \alpha S_v\left(D(G)A_\alpha(G)\right) + \alpha(1-\alpha)S_v\left(A(G)D(G)\right) + (1-\alpha)^2S_v\left(A^2(G)\right)\\
        & = \alpha d^2(v) + \alpha(1-\alpha)\left(2m - d(v) - \sum_{\substack{u \nsim v\\ u\neq v}}d(u)\right) + (1-\alpha)^2\left(2m - d(v) - \sum_{\substack{u \nsim v\\ u\neq v}}d(u)\right)\\
        & = \alpha d^2(v) + (1-\alpha)\left(2m - d(v) - \sum_{\substack{u \nsim v\\ u\neq v}}d(u)\right)\\
        & \leq \alpha d^2(v) + (1-\alpha)(2m - d(v) - (n - d(v) - 1)\delta)\\
        & = \alpha d^2(v) + (1-\alpha)(2m +(\delta - 1)d(v) - \delta(n - 1))\\
        & \leq \alpha \Delta^2 + (1-\alpha)(2m +(\delta - 1)\Delta - \delta(n - 1))
    \end{align*}
   Hence, for every $v \in V$, we have
    \begin{equation*}
        S_v(A^2_\alpha(G)) \leq \alpha \Delta^2 + (1-\alpha)(2m +(\delta - 1)\Delta - \delta(n - 1)).
    \end{equation*}
    By Lemma \ref{lemma::rowsums},
    \begin{equation*}
        \lambda_1^2(A_\alpha(G)) \leq \alpha \Delta^2 + (1-\alpha)(2m +(\delta - 1)\Delta - \delta(n - 1))
    \end{equation*}
    Solving the quadratic inequality, we obtain the result.

    Suppose initially that $G$ is a $r$-regular graph, so
    \begin{align*}
        \sqrt{\alpha \Delta^2 + (1 - \alpha)(\Delta(\delta - 1) -\delta(n-1) + 2m) } &= \sqrt{\alpha r^2 + (1 - \alpha)\left(r(r - 1) -r(n-1) + 2\dfrac{nr}{2}\right)} \\
        &=\sqrt{\alpha r^2 + (1 - \alpha)\left(r^2 - r -rn +r + nr\right)} = r
    \end{align*}
    Now, suppose that holds the equality in (\ref{eq::UpperBound_JCL}). This implies that all inequalities present in its proof are equalities. So, $d(v) = \Delta$ for every $v \in V$, which implies that $G$ is a regular graph.
\end{proof}

\subsection{Bounds Comparison}

In this subsection, tables and graphs are presented to compare the proposed bounds by the authors in Subsection \ref{subsection::A_alpha-Graphs} among themselves and also to compare these bounds with others found in the literature.

\subsubsection{Lower bounds proposed}
We provide tables illustrating the performance of the two bounds introduced in the preceding section alongside the precise value of $\lambda_1(A_\alpha(G))$. Notice that these bounds have different extremal graphs. Three families of graphs were used, two from special trees, the path and the binomial tree, and the third from the pineapple graph.

Starting with the path $P_n$, in Tab.~\ref{tab::lower_comparison_path} we notice that the values obtained using bound~\eqref{eq::lowerbound} are always better than those obtained by bound~\eqref{eq::lower_bound2}, starting with a notable difference which reduces as $\alpha$ increases, and for values of $\alpha$ close to $1$, it ceases to be significant.

\begin{table}[H]
    \centering
    \resizebox{\textwidth}{!}{
        \begin{tabular}{c|c|c|c|c|c|c|c|c|c|c|c}
        \hline
        \multicolumn{12}{c}{Lower Bound Comparison for $\mathbf{G \cong P_n}$} \\
        \hline
        \multicolumn{2}{c|}{$\alpha$} & 0.0 & 0.1 & 0.2 & 0.3 & 0.4 & 0.5 & 0.6 & 0.7 & 0.8 & 0.9 \\
        \hline
        \multirow{3}{*}{$n = 100$} & $\lambda_1(A_\alpha)$ & 1.99903 & 1.99913 & 1.99922 & 1.99931 & 1.99941 & 1.99951 & 1.9996 & 1.9997 & 1.9998 & 1.9999\\
        & (\ref{eq::lowerbound}) & 1.41421 & 1.42377 & 1.43578 & 1.45125 & 1.47178 & 1.5 & 1.54031 & 1.6 & 1.6899 & 1.8217\\
        & (\ref{eq::lower_bound2}) & 0.07841 & 0.26906 & 0.45971 & 0.65036 & 0.841 & 1.03165 & 1.2223 & 1.41295 & 1.6036 & 1.79425\\
        \hline
        \multirow{3}{*}{$n = 500$} & $\lambda_1(A_\alpha)$ & 1.99996 & 1.99996 & 1.99997 & 1.99997 & 1.99998 & 1.99998 & 1.99998 & 1.99999 & 1.99999 & 2.0\\
        & (\ref{eq::lowerbound}) & 1.41421 & 1.42377 & 1.43578 & 1.45125 & 1.47178 & 1.5 & 1.54031 & 1.6 & 1.6899 & 1.8217\\
        & (\ref{eq::lower_bound2}) & 0.01594 & 0.21404 & 0.41215 & 0.61025 & 0.80836 & 1.00647 & 1.20457 & 1.40268 & 1.60078 & 1.79889\\
        \hline
        \multirow{3}{*}{$n = 1000$} & $\lambda_1(A_\alpha)$ & 1.99999 & 1.99999 & 1.99999 & 1.99999 & 1.99999 & 2.0 & 2.0 & 2.0 & 2.0 & 2.0\\
        & (\ref{eq::lowerbound}) & 1.41421 & 1.42377 & 1.43578 & 1.45125 & 1.47178 & 1.5 & 1.54031 & 1.6 & 1.6899 & 1.8217\\
        & (\ref{eq::lower_bound2}) & 0.00798 & 0.20704 & 0.40609 & 0.60514 & 0.80419 & 1.00324 & 1.20229 & 1.40134 & 1.6004 & 1.79945\\
        \end{tabular}
    }
    \caption{Lower bounds comparison for a path graph increasing the number of nodes.}
    \label{tab::lower_comparison_path}
\end{table}

In Tab.\ref{tab::lower_comparison_BT}, the comparative study was carried out using the binomial tree, $BT_k$. The values obtained by \eqref{eq::lowerbound}  are also better than those obtained by \eqref{eq::lower_bound2}, and in this case, regardless of the value of $\alpha$, the values obtained by the first bound are much better than those obtained by the second.
\begin{table}[H]
    \centering
    \resizebox{\textwidth}{!}{
        \begin{tabular}{c|c|c|c|c|c|c|c|c|c|c|c}
        \hline
        \multicolumn{12}{c}{Lower Bound Comparison for $\mathbf{G \cong BT_k}$} \\
        \hline
        \multicolumn{2}{c|}{$\alpha$} & 0.0 & 0.1 & 0.2 & 0.3 & 0.4 & 0.5 & 0.6 & 0.7 & 0.8 & 0.9 \\
        \hline
        \multirow{3}{4em}{\centering $k = 7$\\ $(n = 128)$} & $\lambda_1(A_\alpha)$ & 3.45291 & 3.62629 & 3.82894 & 4.06529 & 4.3399 & 4.65723 & 5.02142 & 5.436 & 5.9036 & 6.42533\\
        & (\ref{eq::lowerbound}) & 2.64575 & 2.8 & 3.0 & 3.25913 & 3.58997 & 4.0 & 4.48806 & 5.04499 & 5.65764 & 6.31293\\
        & (\ref{eq::lower_bound2}) & 0.06153 & 0.26415 & 0.46677 & 0.66939 & 0.87201 & 1.07463 & 1.27725 & 1.47988 & 1.6825 & 1.88512\\
        \hline
        \multirow{3}{4em}{\centering $k = 9$ \\ $(n = 512)$} & $\lambda_1(A_\alpha)$ & 3.97329 & 4.23334 & 4.53693 & 4.88921 & 5.29487 & 5.75785 & 6.28124 & 6.8671 & 7.51632 & 8.22825\\
        & (\ref{eq::lowerbound}) & 3.0 & 3.22947 & 3.52982 & 3.91868 & 4.40832 & 5.0 & 5.68328 & 6.44109 & 7.25576 & 8.11248\\
        & (\ref{eq::lower_bound2}) & 0.01556 & 0.21852 & 0.42148 & 0.62443 & 0.82739 & 1.03035 & 1.23331 & 1.43626 & 1.63922 & 1.84218\\
        \hline
        \multirow{3}{5em}{\centering $k = 10$ \\ {{$(n = 1024)$}}} & $\lambda_1(A_\alpha)$ & 4.21077 & 4.51674 & 4.87387 & 5.28746 & 5.76201 & 6.30095 & 6.90661 & 7.58018 & 8.32164 & 9.12947\\
        & (\ref{eq::lowerbound}) & 3.16228 & 3.43141 & 3.78514 & 4.24278 & 4.81534 & 5.5 & 6.28161 & 7.13976 & 8.05513 & 9.01233\\
        & (\ref{eq::lower_bound2}) & 0.0078 & 0.21 & 0.41221 & 0.61441 & 0.81662 & 1.01883 & 1.22103 & 1.42324 & 1.62544 & 1.82765\\
        \end{tabular}
    }
    \caption{Lower bounds comparison for a binomial tree graph increasing the number of nodes.}
    \label{tab::lower_comparison_BT}
\end{table}

The question that arose was: would bound \eqref{eq::lowerbound} always produce better approximations than bound \eqref{eq::lower_bound2}? To answer this question, we sought to use a family of graphs that were structurally very different from the tree, having chosen the pineapple graph. As we can see in Tab.\ref{tab::lower_comparison_pineapple}   the situation is reversed. Here the values obtained by \eqref{eq::lowerbound}  are always much worse than those provided by \eqref{eq::lower_bound2}, which is more evident for values of $\alpha$ close to 0.
\begin{table}[H]
    \centering
    \resizebox{\textwidth}{!}{
        \begin{tabular}{c|c|c|c|c|c|c|c|c|c|c|c}
        \hline
        \multicolumn{12}{c}{Lower Bound Comparison for $\mathbf{G \cong K_p^1}$} \\
        \hline
        \multicolumn{2}{c|}{$\alpha$} & 0.0 & 0.1 & 0.2 & 0.3 & 0.4 & 0.5 & 0.6 & 0.7 & 0.8 & 0.9 \\
        \hline
        \multirow{3}{*}{$p = 99$} & $\lambda_1(A_\alpha)$ & 98.0001 & 98.00109 & 98.00209 & 98.00309 & 98.00411 & 98.00513 & 98.00617 & 98.00725 & 98.00842 & 98.00999\\
        & (\ref{eq::lowerbound}) & 9.94987 & 15.20784 & 22.62537 & 31.26653 & 40.48902 & 50.0 & 59.66816 & 69.42964 & 79.25048 & 89.11122\\
        & (\ref{eq::lower_bound2}) &95.13821 & 95.42441 & 95.71061 & 95.99681 & 96.28301 & 96.56921 & 96.85541 & 97.14161 & 97.42781 & 97.71401\\
        \hline
        \multirow{3}{*}{$p = 499$}  & $\lambda_1(A_\alpha)$& 498.0 & 498.0002 & 498.0004 & 498.0006 & 498.0008 & 498.00101 & 498.00121 & 498.00141 & 498.00162 & 498.00184\\
        & (\ref{eq::lowerbound}) & 22.33831 & 57.00312 & 102.90936 & 151.31907 & 200.49776 & 250.0 & 299.66696 & 349.42878 & 399.25009 & 449.11113\\
        & (\ref{eq::lower_bound2}) & 495.02793 & 495.32514 & 495.62234 & 495.91955 & 496.21676 & 496.51397 & 496.81118 & 497.10838 & 497.40559 & 497.7028\\
        \hline
        \multirow{3}{*}{$p = 999$} & $\lambda_1(A_\alpha)$ & 998.0 & 998.0001 & 998.0002 & 998.0003 & 998.0004 & 998.0005 & 998.0006 & 998.0007 & 998.0008 & 998.00091\\
        & (\ref{eq::lowerbound}) & 31.60696 & 107.43866 & 202.95339 & 301.32614 & 400.49888 & 500.0 & 599.66681 & 699.42868 & 799.25005 & 899.11112\\
        & (\ref{eq::lower_bound2}) & 995.01398 & 995.31258 & 995.61119 & 995.90979 & 996.20839 & 996.50699 & 996.80559 & 997.1042 & 997.4028 & 997.7014\\
        \end{tabular}
    }
    \caption{Lower bounds comparison for a pineapple graph $K_p^1$ increasing the number of nodes.}
    \label{tab::lower_comparison_pineapple}
\end{table}

Therefore, taking into account the previous observations, we can conclude that the bounds \eqref{eq::lowerbound} and \eqref{eq::lower_bound2} are incomparable.

\subsubsection{Lower bounds proposed and Nikiforov's bounds.}

Now, we move on to comparing our bounds with other bounds present in literature. Among the variety of lower bounds that exist, we selected those that are due to Nikiforov and that have the same extremal graphs as ours. The choice of these bounds is due to the fact that they were the first to appear in the literature and also as a simple tribute to the importance of Nikiforov's contributions to Spectral Graph Theory. Thus, having the star as an extremal graph, we compare bound~\eqref{eq::lowerbound} with bound~\eqref{eq::niki_lowerBound_extremalStar} and, having the regular graph as an extremal, we compare bounds~\eqref{eq::lower_bound2} and~\eqref{eq::LowerUpperBound_Wang}. 

The results are presented through graphs in the variable $\alpha$. Each graph shows the difference between the value obtained by each of the bounds and $\lambda_1(A_{\alpha}(W_n))$, which makes it easier to observe the aspects to be highlighted.
  
Firstly we note that for regular graphs the bound \eqref{eq::lowerbound} is better than the bound \eqref{eq::niki_lowerBound_extremalStar}, which is proved in the next Proposition.

\begin{proposition} \label{prop::LowerBounds_comparison_for_regular}
    If $G$ is a $r$-regular graph and $\alpha \in [0,1]$, then the lower bound in \eqref{eq::lowerbound} is better than the lower bound in \eqref{eq::niki_lowerBound_extremalStar}.
\end{proposition}

\begin{proof}
    Since $r \in \nat^{*}$ and $\alpha \in [0,1]$ we have
    \begin{align*}
         &\left(r^{\frac{3}{2}} - r^{\frac{1}{2}}\right)  \alpha(1-\alpha)\geq 0\\
         &-\alpha^2r^{\frac{3}{2}} + \alpha^2r^{\frac{1}{2}} + \alpha r^{\frac{3}{2}} -\alpha r^{\frac{1}{2}}  \geq 0 \\
         &-\alpha^2r^{\frac{3}{2}} + \alpha^2r^{\frac{1}{2}} + \alpha r^{\frac{3}{2}} -\alpha r^{\frac{1}{2}} + \dfrac{\alpha^2r^2}{4} + \dfrac{\alpha^2r}{2} + \dfrac{\alpha^2}{4} - 2r\alpha + r - \dfrac{\alpha^2r^2}{4} - \dfrac{\alpha^2r}{2} - \dfrac{\alpha^2}{4} + 2r\alpha - r  \geq 0 \\
         &-\alpha^2r^{\frac{3}{2}} + \alpha^2r^{\frac{1}{2}} + \alpha r^{\frac{3}{2}} -\alpha r^{\frac{1}{2}} + \dfrac{\alpha^2r^2}{4} + \dfrac{\alpha^2r}{2} + \dfrac{\alpha^2}{4} - 2r\alpha + r \geq + \dfrac{\alpha^2r^2}{4} + \dfrac{\alpha^2r}{2} + \dfrac{\alpha^2}{4} - 2r\alpha + r \\
         &\left( \dfrac{2r\alpha + 2(1-\alpha)r^{\frac{1}{2}} - \alpha(r+1)}{2}\right)^2 \geq \left( \dfrac{\sqrt{\alpha^2(r+1)^2 + 4r(1-2\alpha)}}{2}\right)^2 
    \end{align*}
    So,
    \begin{equation*}
        \dfrac{2r\alpha + 2(1-\alpha)r^{\frac{1}{2}} - \alpha(r+1)}{2} \geq \dfrac{\sqrt{\alpha^2(r+1)^2 + 4r(1-2\alpha)}}{2}
    \end{equation*}
    and then the result follows.
\end{proof}

For the graph wheel $(W_n)$, which is not a regular graph, Figure \ref{fig::lower_wheel} shows that the bounds \eqref{eq::lowerbound} and \eqref{eq::niki_lowerBound_extremalStar} have a similar behavior, although bound \eqref{eq::lowerbound} is slightly better. Note that, the greater the number of vertices, the more similar the values provided by the two bounds.

\begin{figure}[H]
\centering
    \includegraphics[width=0.8\textwidth]{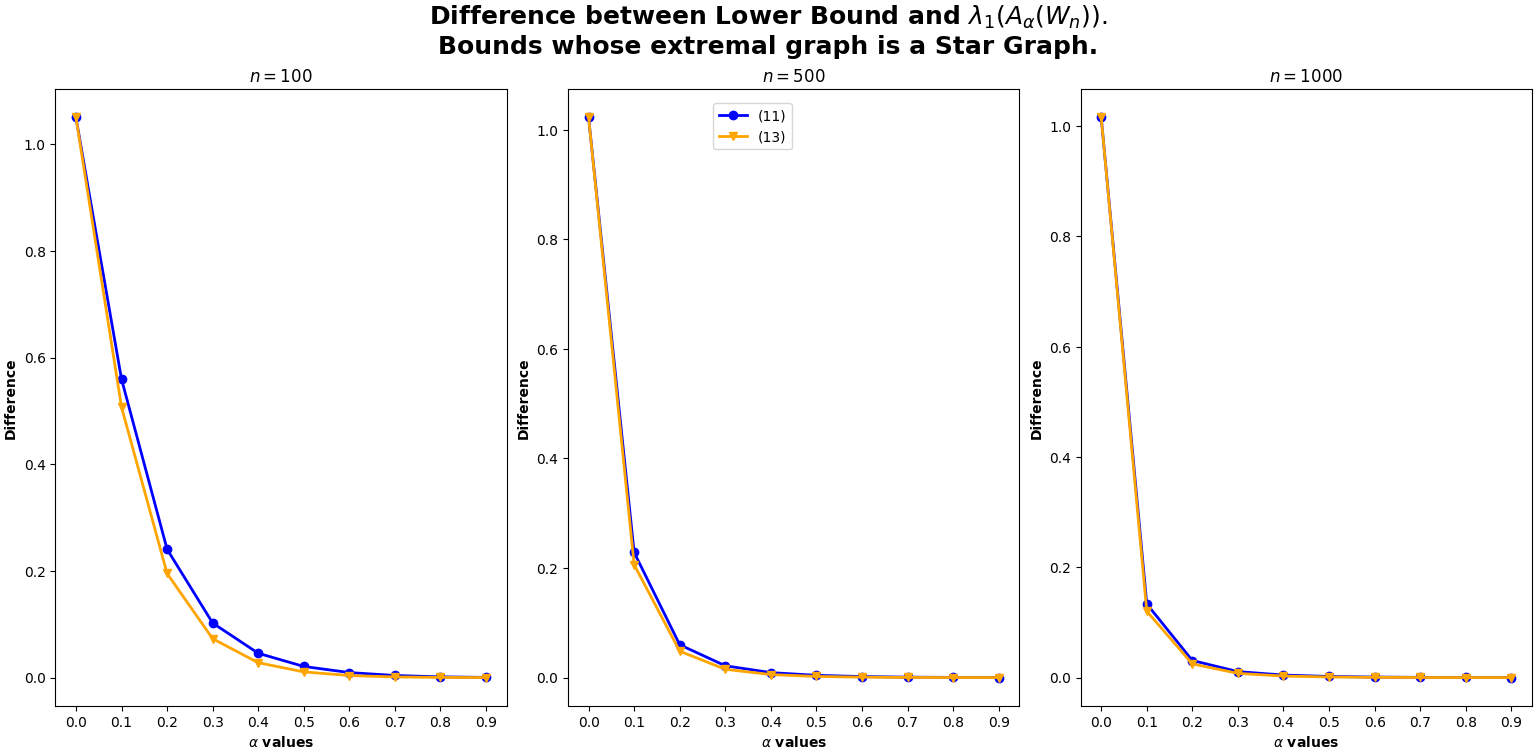}
    \caption{Comparing lower bound for $W_n$.}
    \label{fig::lower_wheel}
\end{figure}

From Corollary~\ref{cor::LowerBoundEquality}, Proposition~\ref{prop::LowerBounds_comparison_for_regular} and after some computational tests, the following question arises: is the bound \eqref{eq::lowerbound} always better than the bound \eqref{eq::niki_lowerBound_extremalStar}? It is worth noting that even in light of the results, this question is still under discussion.

Now, we see what happens with the bounds  \eqref{eq::LowerUpperBound_Wang} and \eqref{eq::lower_bound2}. If we consider $G \simeq K_{1,n-1}$ or $G \simeq W_n$, we can observe from Figures \ref{fig::lower_star} and \ref{fig::lower_wheel1}, respectively,  that for some values of $\alpha$,  Nikiforov's bound is better than our bound, and that for the other values of $\alpha$, the opposite happens. Therefore, these bounds are incomparable.

\begin{figure}[H]
\centering
    \includegraphics[width=0.8\textwidth]{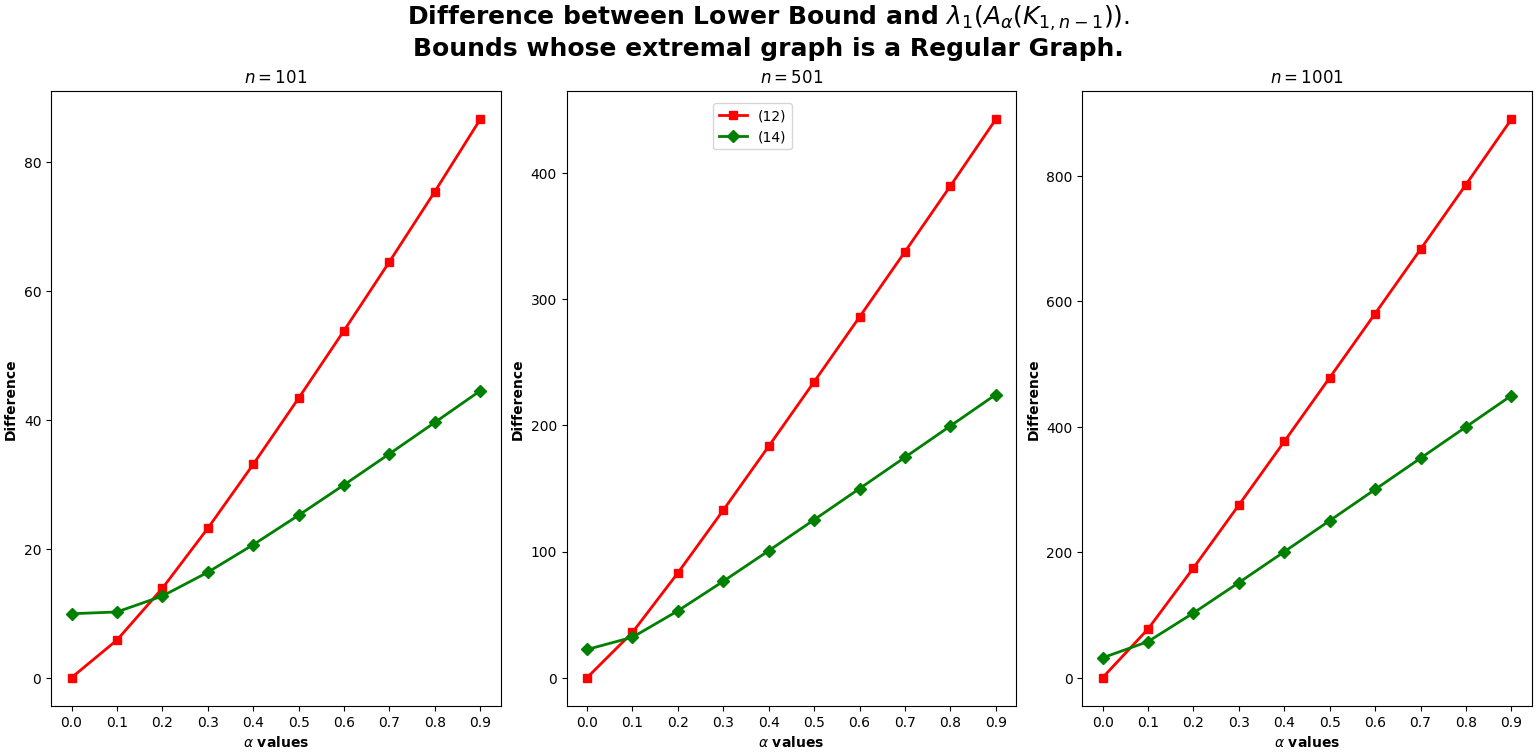}
    \caption{Comparing lower bound for $K_{1,n-1}$.}
    \label{fig::lower_star}
\end{figure}

\begin{figure}[H]
\centering
    \includegraphics[width=0.8\textwidth]{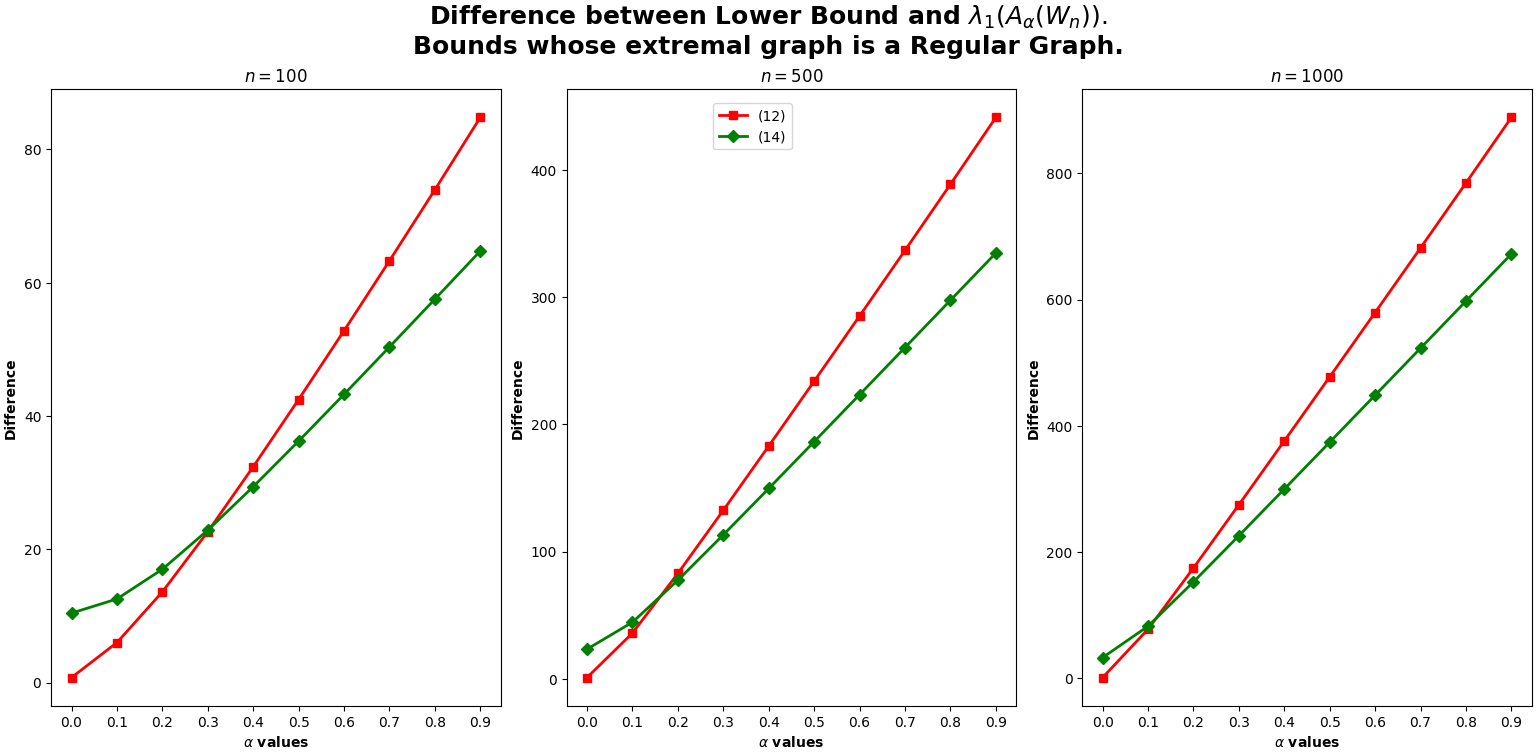}
    \caption{Comparing lower bound for $W_n$.}
    \label{fig::lower_wheel1}
\end{figure}

\subsubsection{Upper bounds}

Now, utilize the upper bounds, we compare the bounds introduced in \eqref{eq::LowerUpperBound_Wang} and \eqref{eq::UpperBound_JCL}. The Proposition~\ref{prop::UpperBound_ineq} shows that the bound \eqref{eq::LowerUpperBound_Wang} is always lower or equal than the bound \eqref{eq::UpperBound_JCL}. 

\begin{proposition} \label{prop::UpperBound_ineq}
    Let $G$ be a graph with $n$ vertices, $m \neq 0$ edges, $\Delta$ its maximum degree, $\delta$ its minimum degree and $\alpha \in [0,1]$. Then, 
    \begin{equation} \label{eq::UpperBound_ineq}
        \max_{v_i \in V}\left\{ \sqrt{\alpha d^2(v_i) + (1-\alpha)\sum_{v_j \sim v_i}d(v_j)} \right\} \leq \sqrt{\alpha \Delta^2 + (1 - \alpha)(\Delta(\delta - 1) -\delta(n-1) + 2m) }
    \end{equation}
    Moreover, if equality holds then the maximum in the left side of \eqref{eq::UpperBound_ineq} is achieved when 
    $d(v_i) =  \Delta$. 
\end{proposition}
\begin{proof}
    Suppose that the maximum of the left-hand side of the inequality \eqref{eq::UpperBound_ineq} is reached at vertex $v_i$ and to simplify the notation, we use $d(v_i) = d_i$ for all $i \in \{1, \ldots, n\}$. Now,
    
    \begin{align}
        &\alpha \Delta^2 + (1 - \alpha)(\Delta(\delta - 1) -\delta(n-1) + 2m) - \alpha d_i^2 - (1-\alpha)\sum_{v_j \sim v_i}d_j = \nonumber\\
        &\alpha(\Delta^2 - d_i^2) + (1-\alpha)\left(\Delta(\delta - 1) -\delta(n-1) + 2m - \sum_{v_j \sim v_i}d_j \right)= \nonumber\\
        &\alpha(\Delta^2 - d_i^2) + (1-\alpha)\left(\Delta(\delta - 1) -\delta(n-1) + \sum_{1\leq s\leq n}d_s - \sum_{v_j \sim v_i}d_j \right) =\nonumber \\
        &\alpha(\Delta^2 - d_i^2) + (1-\alpha)\left(\Delta(\delta - 1) -\delta(n-1) + \sum_{v_j \nsim v_i}d_j + d_i \right)= \nonumber 
        \end{align}

        \begin{align}
        &\alpha(\Delta^2 - d_i^2) + (1-\alpha)\left(\delta(\Delta - n+1) + \sum_{v_j \nsim v_i}d_j + d_i - \Delta \right) \geq \nonumber \\
        &\alpha(\Delta^2 - d_i^2) + (1-\alpha)\left(\delta(\Delta - n+1) + \delta(n-1-d_i) + d_i - \Delta \right) = \nonumber\\
        &\alpha(\Delta^2 - d_i^2) + (1-\alpha)\left(\delta(\Delta - d_i) - (\Delta - d_i)\right) = \nonumber \\
        &\alpha(\Delta^2 - d_i^2) + (1-\alpha)(\delta - 1)(\Delta - d_i) \label{eq::positive_verification}
    \end{align}
    Here we have two cases to consider: if $d_i \neq \Delta$, follows that \eqref{eq::positive_verification} is greater or equal to $0$; if $d_i = \Delta$, follows that \eqref{eq::positive_verification} is equal to $0$. So, in both cases, the result follows.

    Now, suppose that
    \begin{equation*}
        \alpha \Delta^2 + (1 - \alpha)(\Delta(\delta - 1) -\delta(n-1) + 2m) = \alpha d_i^2 + (1-\alpha)\sum_{v_j \sim v_i}d_j
    \end{equation*}
    for all $\alpha \in [0,1].$ Considering both members of the previous equality  as polynomials in $\alpha$ and taking into account the equality of polynomials, we have  that
    \begin{equation}\label{eq::equality_verification1}
        \Delta^2 - d_i^2  - \sum_{v_j \sim v_i}d_j + \Delta(\delta - 1) -\delta(n-1) + 2m = 0   
    \end{equation}
    and
    \begin{equation}\label{eq::equality_verification2}
        \Delta(\delta - 1) -\delta(n-1) + 2m - \sum_{v_j \sim v_i}d_j = 0.
    \end{equation}
    From equality \eqref{eq::equality_verification2} we obtain 
    \begin{equation*}
        \displaystyle \Delta(\delta - 1) -\delta(n-1) + \sum_{v_j \nsim v_i}d_j + d_i= 0.
    \end{equation*}
    As  $\displaystyle\sum_{v_j \nsim v_i}d_j\geq \delta(n-1-d_i)$ and substituting in the previous equality, we get    
    \begin{equation*}
        0 \geq \Delta(\delta - 1) -\delta(n-1) +\delta(n-1-d_i) + d_i
    \end{equation*}
    and then, after some algebraic manipulation, we get
    \begin{equation*}
        0 \geq (\delta-1)(\Delta - d_i)
    \end{equation*}
    As $\delta \geq 1$ and $\Delta \geq d_i$ we must have that
    \begin{equation*}
         (\delta-1)(\Delta - d_i) = 0,
    \end{equation*}
    and then $\delta = 1$ or $d_i = \Delta$.

    To satisfy both equations, \eqref{eq::equality_verification1} and \eqref{eq::equality_verification2}, we need to have $d_i = \Delta$.
\end{proof}

Although the values obtained from bound \eqref{eq::LowerUpperBound_Wang}  are less or equal than those obtained from bound \eqref{eq::UpperBound_JCL}, in cases of equality, the latter is computed much faster than the former. We reached this conclusion after some computational tests made with, among others, the following families of graphs: $K_{1,n-1}$, $W_n$, regular graphs, $BT_k$, Helm graphs\footnote{A Helm graph, $H_n$, is constructed from a $W_n$ by adding $n$ vertices of degree $1$, one adjacent to each terminal vertex. For more details, we suggest \cite{marr2014magic}.}, Windmill graphs\footnote{A windmill graph $W(\nu, k)$ consists of $\nu$ copies of the complete graph $K_k$, with every vertices connected to a common vertex.}. We can observe the  elapsed time to compute the aforementioned value in Figure~\ref{fig::comparing_time}. 

\begin{figure}[H]
\centering
\subfloat[]{\includegraphics[scale=0.4]{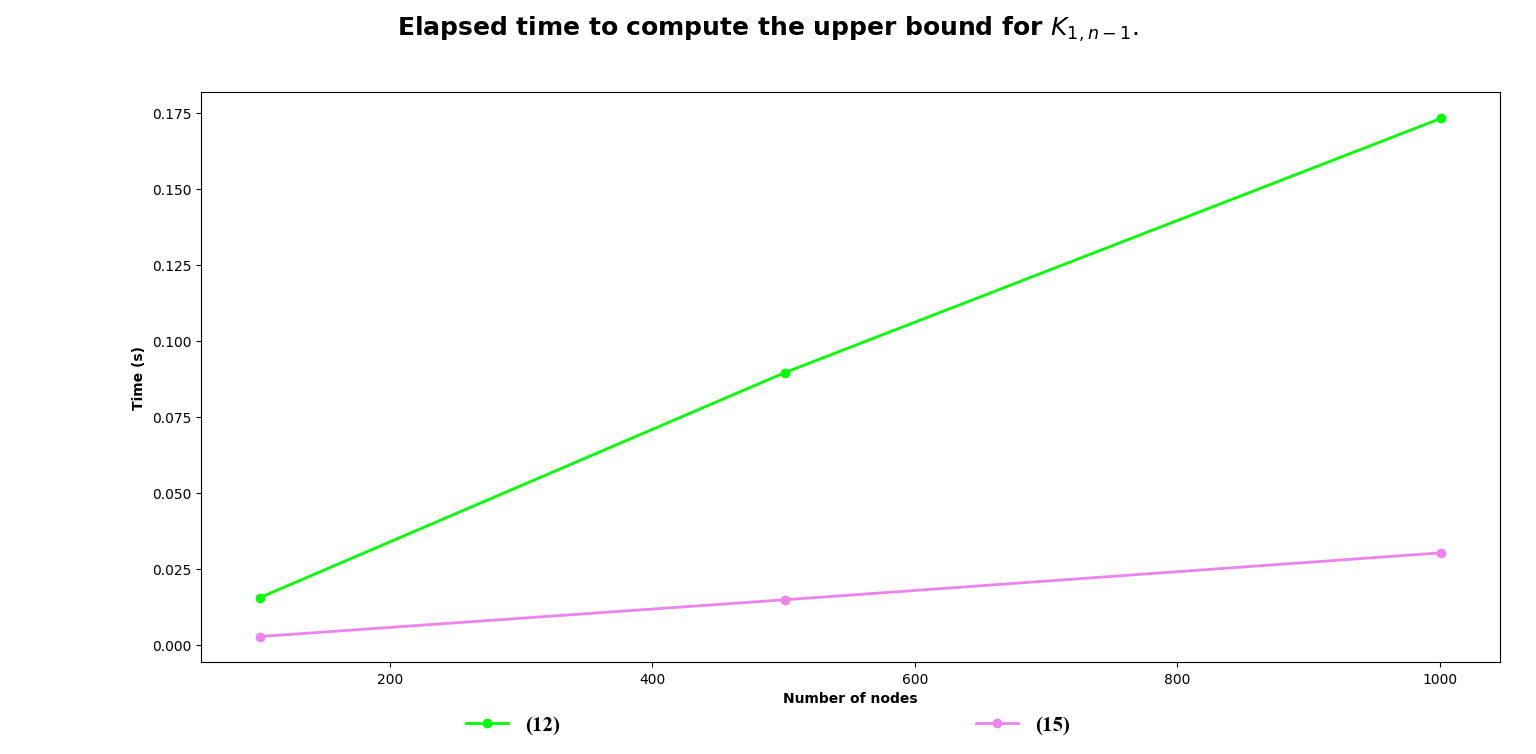}\label{fig::elapsed_star}}\\
\subfloat[]{\includegraphics[scale=0.4]{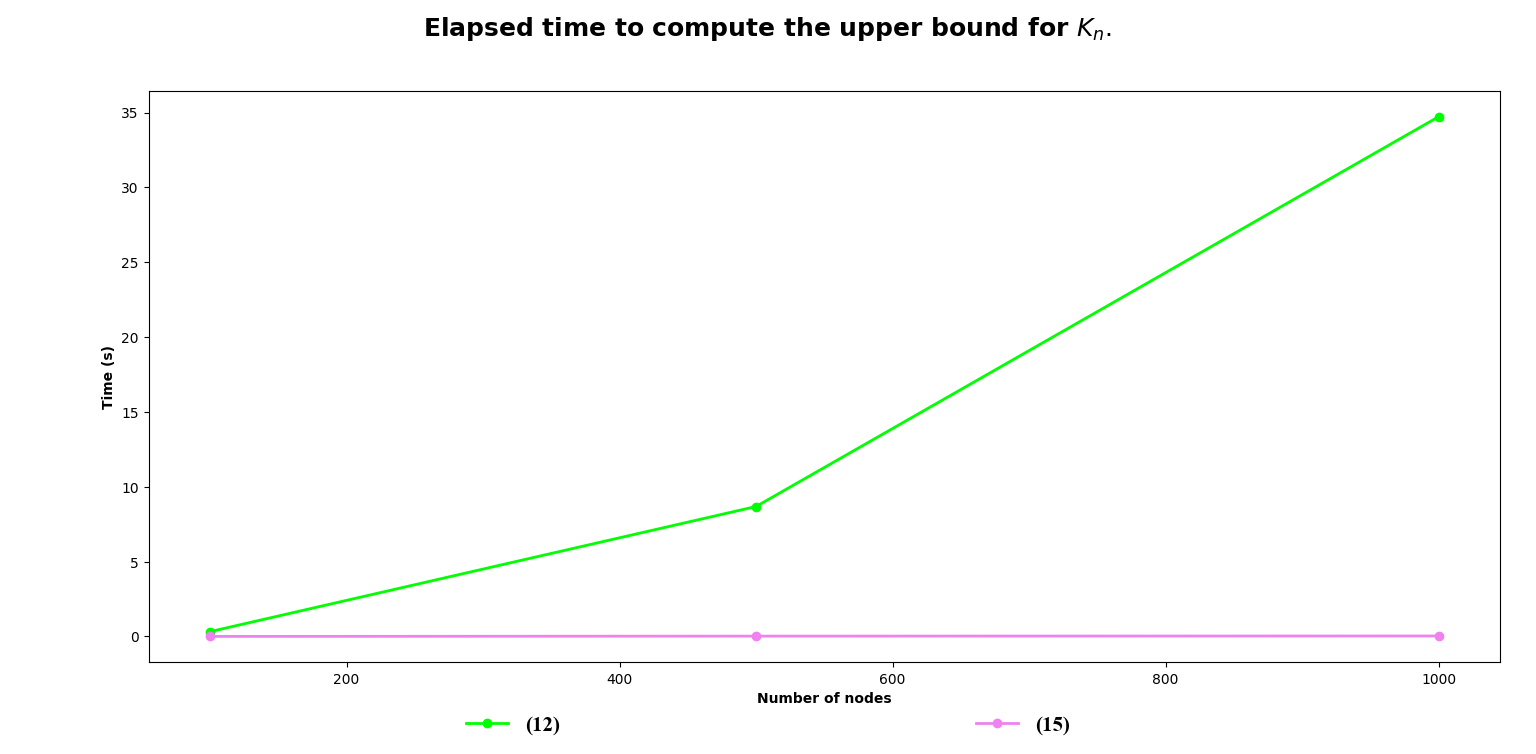}\label{fig::elapsed_complete}}\\
\end{figure}

\begin{figure}[H]
  \ContinuedFloat 
  \centering
\subfloat[]{\includegraphics[scale=0.4]{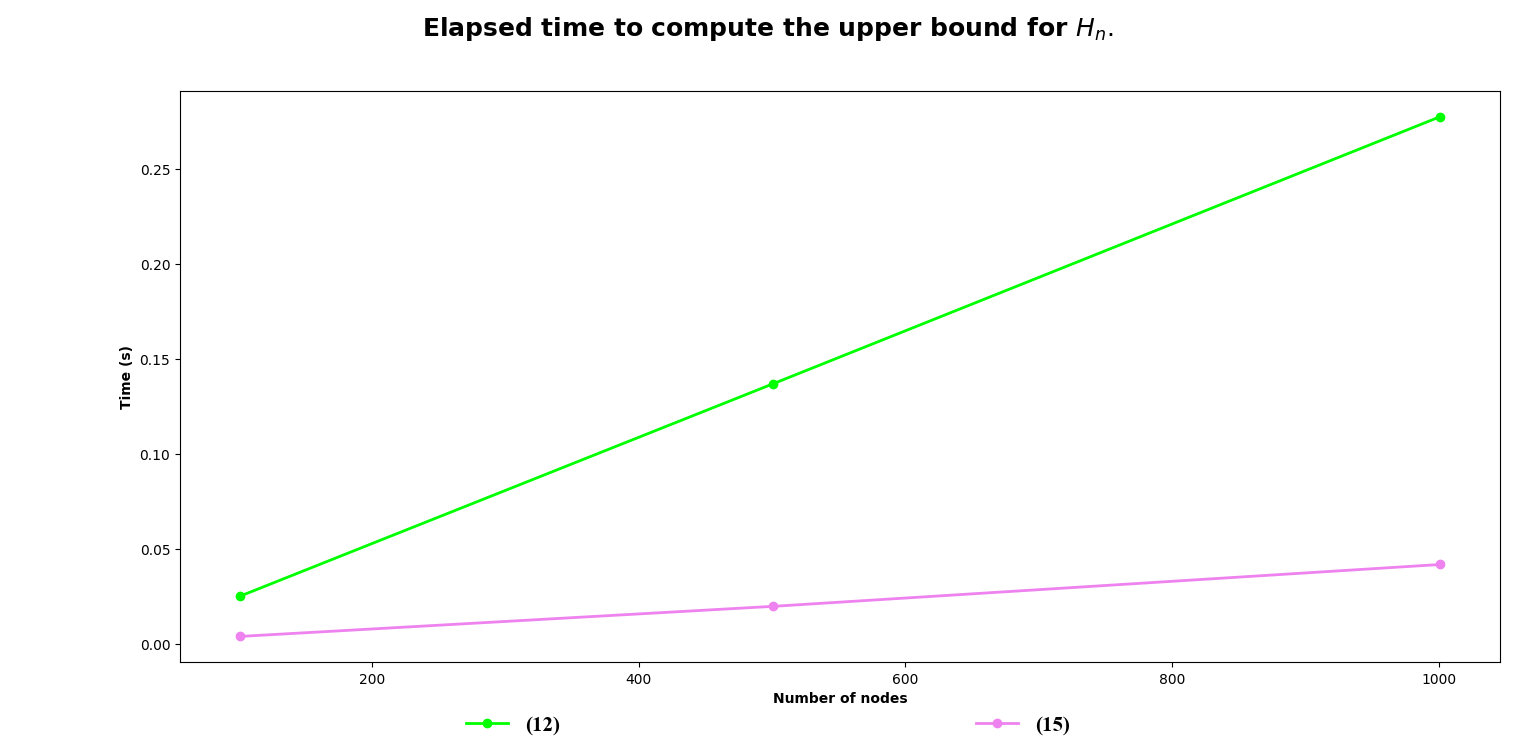}\label{fig::elapsed_helm}}
\caption{Average time spent calculating the upper bound $100$ times for each number of different vertices in each graph.}
\label{fig::comparing_time}
\end{figure}

\subsection{Some results for \texorpdfstring{$\mathbf{A_\alpha}$}{text2}-spectra of Line Graphs} \label{subsection::Graphs}

In this section we present some results involving line graphs. The study of the $A_\alpha$ matrix is very recent and, after a bibliographic research, no results were found, at least so far, involving line graphs and this matrix.  Next theorem is based on Theorem \ref{theo::lowerBound_LG} and presents a lower bound for the smallest eigenvalue of $A_\alpha(l(G))$.

\begin{theorem} \label{theo::lower_bound_lambda_n}
Let $G$ be a graph with $m$ edges and $\delta(G)$ its minimum degree. Then,

\begin{equation} \label{theo::min_eigen_lower_bound}
\lambda_m(A_\alpha(l(G))) \geq 2\alpha \delta(G) - 2,
\end{equation}
for $\alpha \in [0,1)$.
\end{theorem}

\begin{proof}
Let $B$ the incidence matrix of $G$. From Remark \ref{rmk::incident_BTB_alpha} we have
$$(1-\alpha)B^TB = A_\alpha(l(G)) + 2(1-\alpha)I_m - \alpha D(l(G)).$$
Let $\{w_1, \ldots , w_m\}$ be the  orthonormal basis of eigenvectors associated to $\lambda_1(A_\alpha(l(G))), \ldots,\lambda_m(A_\alpha(l(G)))$ such that $A_\alpha(l(G))w_i = \lambda_i(A_\alpha(l(G)))w_i, \ \ \forall i = 1, \ldots,m$.  Then
\begin{align*}
    (1-\alpha)w^T_mB^TBw_m &= w^T_mA_\alpha(l(G))w_m + 2(1-\alpha)w^T_mw_m - \alpha w^T_mD(l(G))w_m\\
    & \leq \lambda_m(A_\alpha(l(G))) + 2(1-\alpha) -\alpha \delta(l(G)).
\end{align*}
As $B^TB$ is positive semi-definite, we get

$$0 \leq w^T_mB^TBw_m \leq \dfrac{\lambda_m(A_\alpha(l(G))) + 2(1-\alpha)-\alpha \delta(l(G))}{1-\alpha}.$$
Therefore,
\begin{equation*}
    0 \leq \lambda_m(A_\alpha(l(G))) + 2(1-\alpha)-\alpha \delta(l(G))
\end{equation*}

From Theorem \ref{theo::degree_relation} the result follows
$$  \lambda_m(A_\alpha(l(G))) \geq 2\alpha \delta(G) -2$$
\end{proof}

\begin{corollary} \label{cor::lowerBound}
If $G$ is a $r$-regular graph with $n$ vertices, $m \neq 0$ edges, $r \geq 2$ and $\alpha \in [0,1)$,  then $\lambda_m(A_\alpha(l(G))) = 2r\alpha -2$ if and only if $m > n$ or $G$ is bipartite.
\end{corollary}
\begin{proof}
Let $G$ a $r$-regular graph. Firstly suppose that $\lambda_m(A_\alpha(l(G))) = 2r\alpha -2$. From Corollary~\ref{cor::pol_signless}, we have  that 
$m > n$ or $0$ is root of $P_{Q(G)}$ and, from Proposition~\ref{prop::signlessLeasteigen}, it follows that $G$ is bipartite. 

If $m > n$,  from Theorem~\ref{theo::linegraph} we have 
that $2r\alpha -2$ is an eigenvalue of $ A_\alpha(l(G))$ and from Theorem~\ref{theo::lower_bound_lambda_n} it is the smallest one. If $G$ is bipartite,  from Proposition \ref{prop::signlessLeasteigen}, we have that $0 \in \sigma(Q(G))$ and applying Corollary~\ref{cor::pol_signless}, the result follows.
\end{proof}

\begin{example}
If $G \cong C_n$ and $n$ is even, from Corollary \ref{cor::lowerBound} we get that  $\lambda_m(A_\alpha(l(G))) = 4\alpha - 2.$
\end{example}

Next propositions present some bounds for the largest eigenvalue of $A_\alpha(l(G))$. 

\begin{proposition}
Let $G$ be a connected graph with $n$ vertices and $\alpha \in [0,1)$. Then, $\lambda_1(A_\alpha(l(G))) < 2$ if and only if $G\cong P_n$.
\end{proposition}
\begin{proof}
Suppose that $G \cong P_n$. From Proposition \ref{prop::upper_largesteigen_Pn}, 
$$\lambda_1(A_\alpha(l(P_n))) = \lambda_1(A_\alpha(P_{n-1})) < 2.$$
Now, suppose by contradiction that $G \ncong P_n$. Then $l(G)$ contains at least a cycle $C$. From Theorem~\ref{theo::interlacing} and Lemma~\ref{lemma::eigeneq_RegularGraphs} we have
$$\lambda_1(A_\alpha(l(G))) \geq \lambda_1(A_\alpha(C)) = 2,$$
and the result follows.
\end{proof}

\begin{theorem} \label{theo::linegraph_UL}
    Let $G$ be a graph with maximum degree $\Delta$ and minimum degree $\delta$. Then
    \begin{equation} \label{ineq::lower}
        \lambda_1(A_\alpha(G)) + (1-2\alpha)\delta -\max_{v_iv_j \in E(G)}\{2-\alpha(d(v_i)+d(v_j))\} \leq \lambda_1(A_\alpha(l(G)))
    \end{equation}
    and
    \begin{equation} \label{ineq::upper}
        \lambda_1(A_\alpha(l(G))) \leq \lambda_1(A_\alpha(G)) + (1-2\alpha)\Delta -\min_{v_iv_j \in E(G)}\{2-\alpha(d(v_i)+d(v_j))\}
    \end{equation}
\end{theorem}
\begin{proof}
Let $B$ be the incidence matrix of $G$. We know that $\lambda_1(BB^T) = \lambda_1(B^TB).$ Applying Corollary \ref{cor::weyl} in \eqref{eq::BTB_1} we obtain
{\small
\begin{equation} \label{ineq::weyl}
    \lambda_1(A_\alpha(l(G))) + \min_{v_iv_j \in E(G)}\{2-\alpha(d(v_i)+d(v_j))\} \leq \lambda_1(A_\alpha(l(G)) + U) \leq \lambda_1(A_\alpha(l(G))) + \max_{v_iv_j \in E(G)}\{2-\alpha(d(v_i)+d(v_j))\}
\end{equation}
}
Let be $x$ a unit non-negative eigenvector
associated to $\lambda_1(A_\alpha(G))$ and $z$ be a unit non-negative eigenvector associated to $\lambda_1(BB^T)$.
From Theorem \ref{theo::rayleigh}, Proposition \ref{prop::rayleigh_alpha} and Remark \ref{rmk::incident_BBT_alpha} we have
\begin{align}
    (1-\alpha)\lambda_1(BB^T) & = (1-\alpha)\max_{\vert z \vert = 1} z^TBB^Tz \geq (1-\alpha)x^TBB^Tx = x^TA_\alpha(G)x + (1-2\alpha)x^TD(G)x \nonumber \\
    &=\lambda_1(A_\alpha(G)) + (1-2\alpha)\sum_{i=1}^nd(v_i)x_i^2 \geq \lambda_1(A_\alpha(G)) + (1-2\alpha)\delta.
\end{align}
Thus, from (\ref{ineq::weyl})
\begin{equation*}
 \lambda_1(A_\alpha(l(G))) + \max_{v_iv_j \in E(G)}\{2-\alpha(d(v_i)+d(v_j))\} \geq  \lambda_1(A_\alpha(G)) + (1-2\alpha)\delta
\end{equation*}
and, as a consequence, inequality (\ref{ineq::lower}) follows.

On the other hand, again from Theorem \ref{theo::rayleigh}, Proposition \ref{prop::rayleigh_alpha} and Remark \ref{rmk::incident_BBT_alpha} we have
\begin{align}
    \lambda_1(A_\alpha(G)) & = \max_{\vert x \vert = 1} x^TA_\alpha(G)x \geq z^TA_\alpha(G)z =  z^T(1-\alpha)BB^Tz - (1-2\alpha)z^TD(G)z \nonumber \\
    & \geq \lambda_1(A_\alpha(l(G))) + \min_{v_iv_j \in E(G)}\{2-\alpha(d(v_i)+d(v_j))\} - (1-2\alpha)\sum_{i=1}^nd(v_i)z_i^2 \nonumber \\ 
    & \geq \lambda_1(A_\alpha(l(G))) + \min_{v_iv_j \in E(G)}\{2-\alpha(d(v_i)+d(v_j))\} - (1-2\alpha)\Delta.
\end{align}
and the inequality (\ref{ineq::upper}) follows.
\end{proof}

\begin{corollary}
    Let $G$ be a $r$-regular graph and $\alpha \in [0,1]$. Then $\lambda_1(A_\alpha(l(G))) = 2r-2$.
\end{corollary}
\begin{proof}
    From Theorem \ref{theo::linegraph_UL},
    \begin{equation*}
        \lambda_1(A_\alpha(G)) + (1-2\alpha)r - 2 + 2r\alpha \leq \lambda_1(A_\alpha(l(G))) \leq \lambda_1(A_\alpha(G)) + (1-2\alpha)r - 2 + 2r\alpha.
    \end{equation*}
    So, $\lambda_1(A_\alpha(l(G))) =  \lambda_1(A_\alpha(G)) + (1-2\alpha)r - 2 + 2r\alpha$. As $\lambda_1(A_\alpha(G))=r$ the result follows.
\end{proof}

\begin{proposition} \label{prop::bound_eigenvalues}
Let $G$ be a graph with $n$ vertices and $m \neq 0$ edges. Then, $\lambda_1(A_\alpha(l(G))) \leq 2n-4$ and $\lambda_i(A_\alpha(l(G))) \leq n(\alpha+1)-4$, $\forall i = 2, \ldots m$ and $\alpha \in [0,1]$.
\end{proposition}
\begin{proof}
As $G$ is a subgraph of $K_n$, the result follows from Corollary \ref{cor::InterlacingLineGraph} and Example~\ref{example::complete_graph}.
\end{proof}

\begin{corollary}
    Let $G$ be a $r$-regular graph with $n \geq 2$ vertices, $m \neq 0$ edges and $\alpha \in [0,1)$. Then, $\lambda_1(A_\alpha(l(G))) = 2n-4$ if and only if $G \cong K_n$.
\end{corollary}
\begin{proof}
    If $G \cong K_n$, from Example~\ref{example::complete_graph},  $\lambda_1(A_\alpha(l(K_n))) = 2n-4.$ On the other hand, if $\lambda_1(A_\alpha(l(G))) = 2n-4$ we have that $P_{A_\alpha(l(G))}(2n-4) = 0$ and, from Theorem \ref{theo::linegraph}, we obtain that $(2n-2 - 2r \alpha)^{m-n}P_{A_\alpha(G)}(2n-2 - r) = 0$, for all $\alpha \in [0,1)$. Now, consider two cases: 
    \begin{itemize}
        \item[(i)] Suppose  $\; 2n-2 - 2r \alpha = 0$. If $0<\alpha < 1$, then $r = \dfrac{n-1}{\alpha}> n-1$, which is impossible. If $\alpha = 0$, then $ n=1$ which is absurd. 
        \item[(ii)] Suppose $\; P_{A_\alpha(G)}(2n-2 - r) = 0$. As $G$ is a regular graph the multiplicity of $\lambda_1(A_\alpha(G))=r$ is $1$. So we have that $2n-2 - r = r$. Therefore, $r = n-1$, which completes the proof. 
    \end{itemize}
\end{proof}

Finally, next theorem presents bounds for $\lambda_2(A_\alpha(l(G))).$

\begin{theorem} \label{theo::bound_second_largest}
Let $G$ be a connected graph with $n$ vertices, $n \geq 3$. Then
$$ 2\alpha -1 \leq \lambda_2(A_\alpha(l(G))) \leq n(\alpha+1)-4. $$
Equality  occurs when $G \cong P_{3}$, for the lower bound, and when  $G \cong K_n$, for the upper bound.
\end{theorem}

\begin{proof}
As $n \geq 3$, from Proposition~\ref{prop::SubgraphLineGraph} we have that $l(P_3)$ is subgraph of $l(G)$. From Theorem~\ref{theo::interlacing} and Corollary~\ref{cor::StarLineGraph} we have that $2 \alpha - 1 \leq \lambda_2(A_\alpha(l(G)))$ and then we can conclude the lower bound. If $G \cong P_3$, from Corollary~\ref{cor::StarLineGraph}, the equality is  achieved. The upper bound and its equality follows straightforwardly from Proposition \ref{prop::bound_eigenvalues}.
\end{proof}

\section*{Acknowledgments}
The research of C. S. Oliveira is supported by the CNPq Grant 304548/2020-0.

\bibliographystyle{unsrt}  
\bibliography{references}

\end{document}